\documentclass[12pt, onecolumn]{IEEEtran}

\usepackage[latin9]{inputenc}
\usepackage{amssymb}
\usepackage{graphicx}
\usepackage{multicol}
\usepackage{amsmath,epsfig}

\usepackage[ruled,vlined]{algorithm2e}
\DeclareMathOperator\erf{erf}
\usepackage{amsfonts}
\usepackage{amssymb}
\usepackage{amsthm}
\usepackage{bbm}
\usepackage{ifthen}
\usepackage{color}
\usepackage{graphicx}
\usepackage{cite}
\usepackage{float}
\usepackage{enumerate}
\usepackage{amsthm}

\newtheorem{corollary}{Corollary}

\theoremstyle{plain}
\newtheorem{thm}{\protect\theoremname}
\theoremstyle{plain}

\theoremstyle{plain}
\newtheorem{prop}[thm]{\protect\propositionname}
\usepackage{mathtools}

\DeclarePairedDelimiterX\Basics[1](){ #1}

\providecommand{\lemmaname}{Lemma}
\providecommand{\propositionname}{Proposition}
\usepackage{mathtools}
\DeclarePairedDelimiter{\ceil}{\lceil}{\rceil}
\providecommand{\theoremname}{Theorem}
\providecommand{\lemmaname}{Lemma}
\providecommand{\propositionname}{Proposition}
\providecommand{\theoremname}{Theorem}
\makeatother
\providecommand{\lemmaname}{Lemma}
\providecommand{\propositionname}{Proposition}
\providecommand{\theoremname}{Theorem}

\usepackage[noend]{algpseudocode}
\makeatletter

\usepackage{epsfig}\usepackage{amsfonts}\usepackage{amsthm}
\usepackage{cite}\usepackage{float}

\makeatother

\begin{document}

	\title{Optimization of Wireless Relaying With Flexible UAV-Borne Reflecting Surfaces}
	
	\author{Taniya~Shafique, Hina Tabassum, and Ekram Hossain\thanks{T. Shafique and E. Hossain are with the Department of Electrical and Computer Engineering, University of Manitoba, Canada. H. Tabassum is with the Department of Electrical Engineering and Computer Science, York University, Canada.The work was supported by a Discovery Grant from the Natural Sciences and Engineering Research Council of Canada (NSERC).} 
		
	}
	\maketitle
	\begin{abstract}
		This paper presents a theoretical framework to analyze the performance of integrated unmanned aerial vehicle (UAV)-intelligent reflecting surface (IRS) relaying system in which IRS provides an additional degree of freedom combined with the flexible deployment of full-duplex UAV to enhance communication between ground nodes.  Our framework considers three different transmission modes: {\bf (i)} UAV-only mode, {\bf (ii)} IRS-only mode, and {\bf (iii)} integrated UAV-IRS mode to achieve spectral and energy-efficient relaying. For the proposed modes, we provide exact and approximate expressions for the end-to-end outage probability, ergodic capacity, and energy efficiency (EE) in closed-form. 
		We use the derived expressions to optimize key system parameters such as the UAV altitude and the number of elements on the IRS considering different modes. We formulate the problems in the form of fractional programming (e.g. single ratio, sum of multiple ratios or maximization-minimization of  ratios) and devise optimal algorithms using quadratic transformations.  Furthermore, we derive an analytic criterion to optimally select different transmission modes to maximize ergodic capacity and EE for a given number of IRS elements. Numerical results validate the derived expressions with Monte-Carlo simulations and the proposed optimization algorithms with the solutions obtained through exhaustive search. Insights are drawn related to the different communication modes, optimal number  of IRS elements, and optimal UAV height.
	\end{abstract}

	\section{Introduction}
	\label{sec:Intro}
	
	Intelligent reflecting surfaces (IRS)  are emerging as a key enabling technique to smartly reconfigure wireless propagation environment in  beyond 5G wireless networks \cite{wu2019towards}. The IRS consists of multiple small meta-surfaces that are also referred to as \textit{IRS elements}. IRS enables smart reconfiguration via software-controlled reflections and is energy-efficient since meta-surfaces contain low-cost polymer diode/switch and conductive square patches~\cite{huang2019reconfigurable}  \cite{liaskos2018new}.    The comprehensive intelligent functionality of each element includes reflection, refraction, transmittance and absorption \cite{wu2019towards,wu2019intelligent}. The functionalities can be used all together or in separate based on the application requirement. In contrast to conventional relays that require active transmission and reception, the IRSs do not require any additional radio channel/frequency for signal transmission or reception which makes IRS cost-effective. 
	
	Unlike conventional IRS relaying, integrating IRS with the unmanned aerial vehicles (UAVs) allow flexible deployment of metasurfaces while minimizing the on-board UAV energy consumption
	\cite{vinogradov2019tutorial,mozaffari2019tutorial}. The proactive placement of integrated UAV-IRS system  offers a cost-effective solution with minimal energy consumption and reduced network-wide spectrum resources. 
	In this paper, we consider the mathematical performance  characterization and optimization of an integrated UAV-IRS system.
	
	\subsection{Background Work}
	A series of research works \cite{ji2019performance,zhou2018coverage,yuan2018capacity,kim2018outage} considered signal-to-noise ratio (SNR) outage characterization of UAV-assisted relaying  assuming either line-of-sight (LoS) Rician or non-LoS (NLoS)  Nakagami-$m$ faded aerial channels.  The derived expressions are generally in the form of complicated mathematical functions that cannot be directly used for network planning and optimization purposes. For instance, \cite{ji2019performance,zhou2018coverage} provided closed-form expressions for the SNR outage probability assuming Nakagami-$m$ faded aerial channels with no notion of LoS and NLoS transmissions. In addition, the authors in \cite{yuan2018capacity} assumed Rician-faded LoS aerial channels and derived the SNR outage in the form of Marcum $Q$-function. In \cite{kim2018outage}, the SNR outage probability was analyzed for Rician and Rayleigh fading channels considering LoS and NLoS channels, respectively.  The aforementioned research works \cite{ji2019performance,zhou2018coverage,yuan2018capacity,kim2018outage} overlooked the impact of limited on-board energy of the UAV as well as the circuit and hovering power consumption of the UAV.

	Another series of research works that focused on the energy efficiency maximization of UAV-enabled relaying networks include \cite{ahmed2016energy,zeng2017energy, yang2018energyFixedWingTradeoff,chakareski2019energy,lyu2017placement,alzenad20173,lyu2016cyclical, shafique2019end,yang2019energy}. These research works are solely based on numerical optimization techniques.
	Very recently, we developed a mathematical framework  to characterize   the reliability, energy efficiency, and coverage probability in a UAV-assisted data ferrying network considering Rician-faded aerial channels \cite{shafique2019end}. Using the derived expressions, we optimized the UAV data ferrying distance in three different problem settings, (i) minimize the energy consumption under the constraint of outage probability, (ii) minimize the outage probability under the constraint of energy consumption, and (iii) minimize both the outage probability and energy consumption by considering multi-objective optimization \cite{shafique2019end}. The aforementioned research works did not consider the IRS-assisted UAV systems.
	
	
	To date, a number of  research works considered the statistical performance characterization or optimization of IRS-assisted wireless networks either without UAV \cite{basar2019wireless,huang2019reconfigurable,huang2018energy, bjornson2019intelligent} or with UAV \cite{li2020reconfigurable}. 
	A pioneering  effort to characterize an upper bound on the average symbol error probability has been undertaken in \cite{basar2019wireless}. The research work considered Rayleigh fading channels and simplified the instantaneous SNR given the optimal phase shifts for IRS. The energy efficiency of the system was not considered. 
	A number of research works \cite{huang2019reconfigurable,huang2018energy} focused on maximizing the energy efficiency by optimizing the IRS phase shifts with infinite and low phase resolution capability. 
	An interesting research work is \cite{bjornson2019intelligent} where the authors compared the performance of decode and forward (DF) relaying and IRS-assisted transmission.  The IRS and DF relay were placed in the same fixed location. 
	They also considered maximal ratio combining between the direct and IRS assisted link. Nevertheless, the channel gain coefficients were assumed to be perfectly known.
	The authors in \cite{li2020reconfigurable} considered an IRS to facilitate the transmission between a  mobile UAV and a ground user. The UAV-to-IRS transmission link was modeled as LoS Rician fading channel whereas IRS-to-ground user  link was modeled as NLoS Rayleigh fading channel.  The authors maximized the rate by optimizing IRS phase shifts and the trajectory through numerical optimization considering known channel state information (CSI). Finally, in \cite{zhang2019reflections}, the authors optimized the location of the integrated IRS-UAV system using reinforcement learning approach. 
	
	\subsection{Paper Contribution and Organization}
	Except \cite{basar2019wireless}, most of the aforementioned research works are focused on the optimization of the phase-shifts in IRS-assisted networks using numerical optimization techniques. Furthermore, the performance characterization and optimization of integrated UAV-IRS system have not been  investigated yet.
	
	This paper develops a comprehensive mathematical framework  to characterize the performance of an integrated UAV-IRS system and optimize critical network parameters such as the number of IRS elements and  UAV altitude to maximize the spectral and energy efficiency. Note that IRS micro-controller can perform the optimal switching of IRS elements out of all elements therefore optimal $N$ can be realized in practice. At this point, it is noteworthy that maximization of energy efficiency  and optimization of the number of IRS elements $(N)$ in an integrated UAV-IRS system is crucial due to two reasons: \textbf{(i}) given the limited UAV size, the number of IRS elements that can be deployed on a UAV is limited\footnote{The size of one IRS element is typically in the range  $\lambda/10- \lambda/5$  \cite{liaskos2018new}, where $\lambda$ denotes the wavelength of the transmitted wave. As such, this limitation becomes more evident in low frequencies.}, and  \textbf{(ii)} due to the power consumption associated with each IRS element. Although the power consumption of each IRS element is low, the overall power consumption may become significant for a large number of active IRS elements depending on the {\em phase resolution power consumption} $P_r(b)$, which depends on the number of bits assigned to resolve the phases in an IRS element. For instance, $P_r(b)=5$dBm for 1-bit resolution and $P_r(b)=45$dBm for infinite resolution \cite{huang2018energy}. The $P_r(b)$ depends on the operating frequency and the type of power amplifier  \cite{mendez2016hybrid}. 
	
	\begin{itemize}
		\item We  characterize the outage probability, ergodic capacity, and energy efficiency in an integrated UAV-IRS system (where IRS surface is mounted on the UAV) considering  three different modes, (i) {\em UAV-only mode}, where the UAV performs relaying in full-duplex mode, (ii) {\em IRS-only mode}, where the IRS performs relaying which is implicitly a full-duplex transmission without self-interference, and (iii) {\em Integrated UAV-IRS mode},  where both the UAV and IRS perform relaying and the receiver uses selection combining (SC). The considered model captures the LoS air-to-ground (AtG) Rician fading channels and power consumption of UAV and IRS. 
		
		\item We provide approximate expressions to increase the mathematical tractability of the proposed framework for system optimization  purposes. That is, we incorporate the derived expressions (after  some transformations to tractable mathematical forms) into the  optimization problems. 
		Numerical results validate the derived expressions with Monte-Carlo simulations.
		\item We formulate a variety of the optimization problems where objective functions have a fractional form for IRS-only mode and UAV-only modes, i.e. (i)  maximize EE to optimize the number of IRS elements, (ii) maximize EE to optimize the height of the IRS, (iii) minimize IRS power consumption to optimize the the number of IRS element and transmission power subject to rate constraints, and (iv)  maximize EE to optimize the height of the UAV. {We solve the aforementioned problems and derive optimal solutions using quadratic transformation as a tool from fractional programming.}
		Closed-form optimal solutions are provided, wherever applicable.
		\item We derive an analytic criterion to optimally select the UAV-only and IRS-only transmission modes to maximize the capacity and EE for a given number of IRS elements. 
		\item Numerical results compare the proposed optimal solutions with the solutions obtained through exhaustive search. We note that, compared to the UAV-only mode, the IRS-only mode is energy efficient at lower altitudes with low to moderate number of active IRS elements, and for larger distances between the UAV  and the source or destination.
	\end{itemize}

	The remainder of the paper is organized as follows. We describe the  system model in Section~\ref{sec:SystemModel}. In Section~\ref{AnalysisSection}, we characterize the end-to-end energy efficiency, the SNR outage probability  and data rate for the considered network modes. In Section~\ref{Sec:EEEachMode}, we propose approximations for erdogic capacity and energy efficiency. In Section~{V}, optimization is performed to maximize energy efficiency for IRS elements and UAV height for transmission modes. Mode selection probability and criteria is proposed in the same section. Then, we present the numerical results in Section~\ref{Sec:Simulation}  before we conclude in Section~\ref{sec:Conclusion}.   
	\section{System Model and Assumptions}
	\label{sec:SystemModel}
	
	\subsection{Spatial Deployment of UAV-IRS system}
	\label{subsec:SpatialDeployment}
	We consider an integrated UAV-IRS network in which a UAV carries a large array of IRS elements to assist communication between source ${\bf S}$ and destination ${\bf D}$ located on the ground. 
	We assume that  there exists no direct link between the ${\bf S}$ and ${\bf D}$. In particular, the IRS reflects the incident signal in the desired direction of destination with minimal power consumption.  In addition, the UAV operates as an independent relay between ${\bf S}$ and  ${\bf D}$ since we assume that the UAV has separate transmit and receive antennas. 
	{In Cartesian coordinates, the locations of ${\bf S}$ and  ${\bf D}$ are denoted as $\textbf{w}_s = (x_s, y_s,0)$, and $ \textbf{w}_d =(x_d, y_d,0)$, respectively (Fig.~\ref{fig:DCFDec2018})}.  We also assume that UAV can be placed at any height $h$ such that  $h\in [h_{\rm min}, h_{\rm max}]$ where $h_{\rm min}$ and $ h_{\rm max}$ are decided by aviation authorities. We denote the UAV coordinate as $\mathbf{w_u}=(x_u,y_u,h)$. In two-dimensional Cartesian coordinates, the location of source, destination, and the UAV can be given by  $\mathbf{z_{s}}=(x_s,y_s)$, $\mathbf{z_{d}}=(x_d,y_d)$, and $\mathbf{z_{u}}=(x_u,y_u)$, respectively.
	\begin{figure}[!t]
		\begin{center}			\includegraphics[height=2.5in]{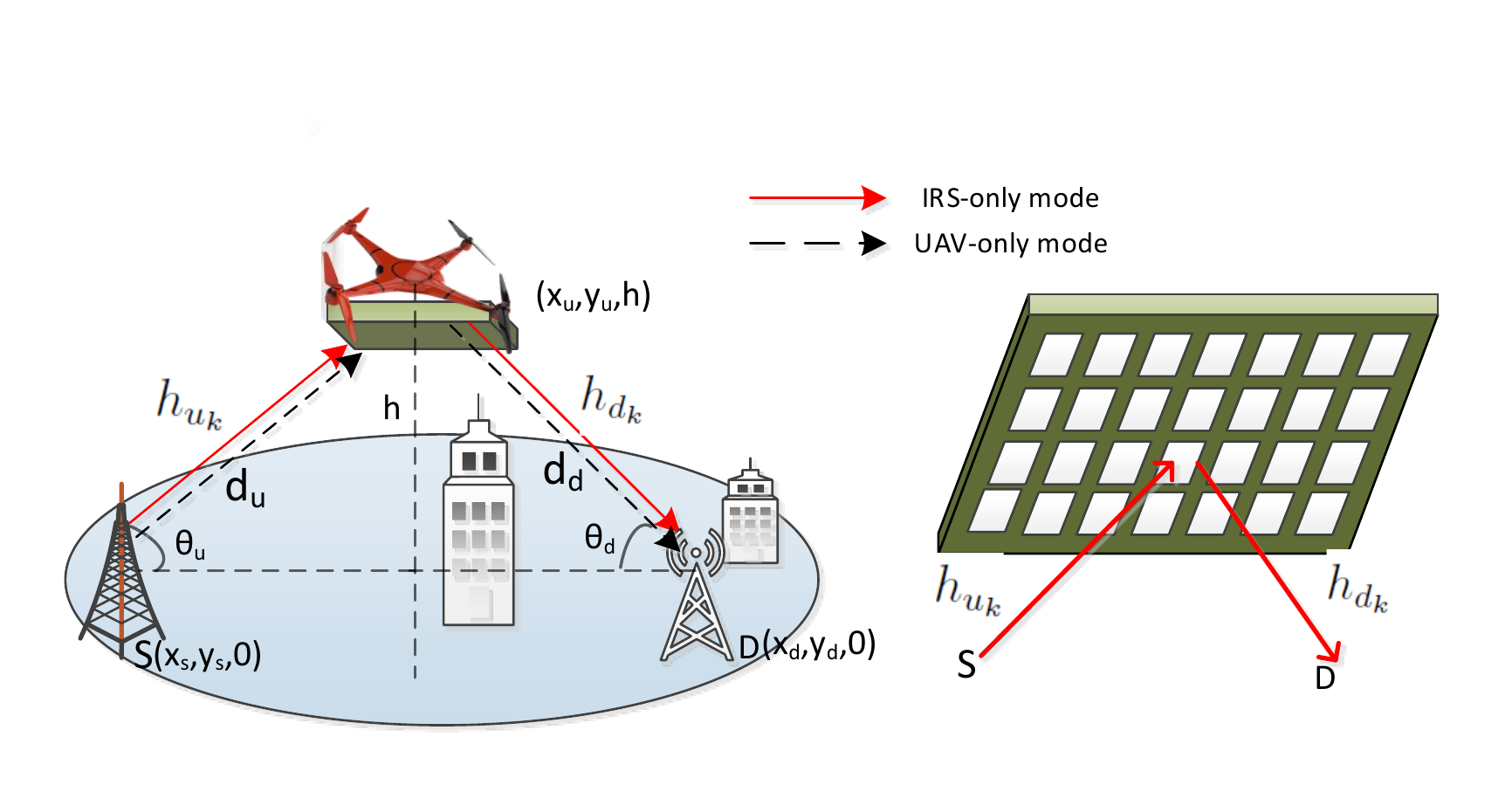}
			\caption{{Integrated UAV-IRS communication} when IRS is placed on a UAV between the source and  the destination.}			\label{fig:DCFDec2018}
		\end{center}
	\end{figure}
	
	\subsection{Aerial Channel Model}
	The communication between the UAV and ground receiver $\bf{S}$ and $ \bf{D}$  depends on the elevation angle between the nodes (and/or altitude of the UAV) and the environment (e.g. the intensity and heights of buildings). The transmission to the ground users may have LoS or non-LOS based on the elevation angle (in rad) between the UAV and  BS$_i$. The elevation angle can be given as follows:
	\begin{align}\label{eq:ThetaAll}
	&\theta_{u}=\arctan\left(\frac{h}{|{{\bf z_u}-{\bf z_s}|}}\right), 
	\quad \quad \quad \;
	\theta_{d}=\arctan\left(\frac{h}{|{\bf z_u}-{\bf z_d}|}\right), 
	\end{align}
	where $h$, ${\bf z_s}$, and ${\bf z_d}$  are defined in \ref{subsec:SpatialDeployment}. The probability of LoS in each link is a function of $\theta_i$, i.e.
	\begin{equation}\label{eq:ProbLoSU}
	p_L(\theta_i)=({1+e_i \mathrm{exp}(-g_i(\theta_i-e_i))})^{-1},\qquad \forall i\in\{u,d\},
	\end{equation} 
	where $e_u$, $e_d$, $g_u$, and $g_d$, are the environment parameters obtained from the curve fitting using Damped Least-Squares (DLS) method \cite{al2018modeling}.
	The path-loss exponent $\alpha$ is a function of the elevation angle~\cite{azari2018ultra}, i.e.
	\begin{align}\label{eq:alphaFuncTheta}
	\alpha(\theta_i) = p_L(\theta_i)q_i+v_i,
	\end{align}
	Here $q_u$, $v_u$, $q_d$, and $v_d$ are constants depending on the uplink and downlink environment \cite{al2018modeling}.

	\subsection{Spectrum Allocation}
	We consider that  the destination BS ${\bf {D}}$ has a data rate requirement $R_{0}$ which is defined as $R_{0}=  B\log_2(1+\Gamma_{0})$.  The $\Gamma_{0}$ represents the minimum end-to-end SNR threshold required by the destination to achieve $R_{0}$, i.e. $\Gamma_{0}={2^{\frac{R_{0}}{B}}-1}$. Here, $B$ represents the total transmission bandwidth available for IRS-only mode, UAV-only mode,  and integrated UAV-IRS mode. The IRS does not need additional frequency to reflect the signals. For the sake of fairness, we consider in-band full-duplex (IBFD) operation for the UAV mode. This enables the UAV to transmit   and receive simultaneously over the same frequency band $B$ which remains the same for all three modes of operation. The performance of IBFD communication is, however, limited by self interference (SI) which is introduced by the IBFD transmitter to its own receiver  \cite{siddique2017downlink}. 
	The antenna is assumed to be equipped with a three-port circulator  to prevent the leakage of transmit chains to receive chains; however, in practice, perfect SI cancellation is not possible \cite{bharadia2013full}. 
	
	\subsection{Transmission Modes}
	We consider three different modes of data transmission, i.e. {\bf(i)} \textit{UAV-only mode}, when UAV provides coverage to the destination $\bf{D}$  with all IRS elements switched off (absorbing state of IRS) and UAV is operating in IBFD transmission mode, {\bf(ii)}  \textit{IRS-only mode}, when only IRS is  responsible to provide service to the destination $\bf {D}$ by acting as relay and the UAV does not communicate, and {\bf(iii)}~\textit{Integrated UAV-IRS mode}, when both IRS and UAV transmit the data and the receiver combines the data using {\em selection combining}\footnote{
		In selection combining, the combiner outputs the signal on the branch with
		the highest SNR, which requires one receiver switching to active branch, and co-phasing of multiple branches is not required as is the case the other combining techniques. Therefore, selection combining exhibits low overhead, has a simplest receiver implementation, and is mathematically tractable.
	}, i.e.  by opportunistically selecting the stronger signal between those received from the UAV  and the IRS. {We consider that the receiver is equipped with a buffer to store the observations from IRS transmission that arrives one time slot prior to the UAV transmission.} 
	We consider that the IRS is equipped with uniform linear arrays of elements and there is a controller associated with IRS which is responsible for smart selection of the functionality of IRS {elements}  such as absorption and {beamforming}.  
	
	Now we describe the transmission and channel models for each of the modes of operation. 
	
	\subsubsection{UAV-only Mode}
	The transmission from  {\bf S} to UAV and the transmission from UAV to {\bf D} can be given, respectively, as follows:
	$$y_{u} ={\sqrt{\hat{A}p_{{u}} \eta_{u}^{-1} \; d_u^{-\alpha(\theta_u)} }}  \;h_{u} \; s+R_{\rm SI}+n_{u}, \qquad
	y_{d}={\sqrt{\hat{A}p_{{d}} \eta_{d}^{-1} \; d_{d}^{-\alpha(\theta_d)} }}  \;h_{d} \; y_{u}+n_d,
	$$ 
	where $s$ is the transmitted signal in binary phase shift keying (BPSK) from  the source $\bf {S}$ to the UAV and $y_u$ is the  signal received by the IBFD UAV and relayed to \textbf{D}, $\eta_{i}$ denotes the excess aerial path-loss, $p_{{u}} $ is the transmission power of $\bf {S}$, and $p_d$ is the transmission power of UAV. {Also, $d_{u}$ is the distance between the $\bf {S}$ and the UAV, i.e. $d_{u}=\sqrt{\vert \textbf{z}_u-\textbf{z}_{s} \vert^2+h^2}$ and  $d_d$ is the distance between UAV and $\bf{D}$, i.e. $d_{d}=\sqrt{\vert \textbf{z}_u-\textbf{z}_{d} \vert^2+h^2}$. }
	Note that $\hat{A}$ reflects system parameters (e.g. operating frequency and antenna gain), $n_i$ is additive white Gaussian noise (AWGN) with zero-mean and power spectral density $N_0$, $R_{\rm SI}$ denotes the residual SI experienced by the UAV \cite{siddique2017downlink}, and $h_{i}$ represents the $i$-th channel fading where $i\in (u,d)$.  The SNR for the $i$-th  is given as follows:
	\begin{equation}
	\gamma_{i}={p_{{i}} \kappa_i\; d_i^{-\alpha(\theta_i)}  \;X_{i}}, \qquad \forall i\in\{u,d\},
	\label{eq:SNRup}
	\end{equation}
	where $\kappa_u=\frac{{\hat{A} \eta_{u}^{-1}}}{R_{\rm SI}+N_0}$,  $\kappa_d=\frac{{\hat{A} \eta_{d}^{-1}}}{N_0}$,
	and $X_{i}=\vert h_{i}\vert^2$ follows non-central chi square distribution with mean $\Omega_i$, which is local mean power of $i$-th Rician fading channel whose probability density function (PDF) is:
	\begin{align}\label{eq:LOSdistribution}
	f_{X_i}(x)& = 
	\frac{K_i+1}{\Omega_i}  e^{-K_i-\frac{(K_i+1)x}{\Omega_i}} I_{0} \left(2\sqrt{\frac{K_i(K_i+1)x}{\Omega_i}}\right)
	=   \sum_{\ell=0}^{\infty} \frac{b_i (b_i K_i)^\ell}{(\ell!)^2} x^\ell e^{-b_i x-K_i},
	\end{align}	
	in which  $K_i$ is the Rician factor in the $i$-th link and $I_0$ is a modified Bessel function of the first kind and $b_i	=\frac{K_i+1}{\Omega_i}$ \cite{nguyen2018two}. Note that in the UAV-only mode, the IRS absorbs the incoming signals to each element, and therefore, no information is relayed from IRS to the destination. We call this state as the non-active state of the IRS.
	
	Assuming that the UAV can perform decoding of $y_u$ and then relay the decoded data, using \eqref{eq:SNRup}, the end-to-end SNR $\Gamma_{\rm UAV}$ from $\bf {S}$ to $\bf {D}$ can be modeled  as~\cite{samarasekera2014performance}:
	\begin{equation} \label{eq:SNRDFMode}
	\Gamma_{\rm UAV}={\min\{\gamma_u,\gamma_{d}\}}.
	\end{equation}
	
	\subsubsection{IRS-only Mode}
	In this mode, we assume that the UAV does not transmit and the IRS controller adjusts the phase shift of each element intelligently to the optimal value \cite{wu2019beamforming,basar2019wireless}. That is, the IRS maximizes the signal power by optimizing the phase shifts of the impinging signals.  In this setup, for the sake of  symmetry, we consider the odd number of elements, i.e. $N=2n+1$,  where $n$ is any arbitrary positive integer. The   received signal at destination {\bf D} via $k$-th IRS element is given by 
	$$y_{\rm IRS_k} ={\sqrt{\hat{A}p_{{u}} \eta_{u}^{-1} \; d_{u_k}^{-\alpha(\theta_{u_k})} }}  \;h_{u_k}e^{j(\phi_k)}\sqrt{ \hat{A}\eta_{d}^{-1} \; {d_{d_k}}^{-\alpha(\theta_{{d_k}})}} h_{d_k}\; s+w_{u},$$ where 
	$k\in\{-n,-n+1,\cdots,0, 1,\cdots, n\}$. The distance between $j$-th element to $\bf{S}$  and to $\bf{D}$ can be given as
	$d_{u_k}=\sqrt{\vert\mathbf{z}_s-\mathbf{z}_{u_k}\vert^2+h^2}$ and 
	$d_{d_k}=\sqrt{\vert\mathbf{z}_d-\mathbf{z}_{u_k}\vert^2+h^2}$, respectively, where  $\mathbf{z}_{u_k}=(x_{u_k},y_{u_k})$ and  $x_{u_k}=x_u- k D_{\rm IRS}$,  $ y_{u_k}=y_{u}$, and $D_{\rm IRS}$ denotes the uniform spacing  between two consecutive elements on IRS. 	Note that the  $k=0$-th element  is at UAV {location} $\mathbf{z}_{u}$.  	The channel from {\bf S} to $k$-th IRS element and $k$-th IRS element to {\bf D} can be given as $h_{u_k}=\vert h_{u_k}\vert e^{-j\theta_{u_k}}$ and $h_{d_k}=\vert h_{d_k}\vert e^{-j\theta_{d_k}}$, respectively.
	The end-to-end SNR for IRS-only mode $\Gamma_{\rm IRS}$ \cite{basar2019wireless} 	for an IRS with $N$ elements can be given  as follows:
	\begin{equation}\label{eq:snrALL}
	\Gamma_{\rm IRS}=V{ \;\left(\sum_{k=1}^{N} d_{u_k}^{-\alpha(\theta_{u_k})/2} d_{d_k}^{-\alpha(\theta_{d_k})/2} \vert h_{u_k}\vert\vert h_{d_k} \vert e^{-j(\theta_{u_k}+\theta_{d_k}-\phi_k)} \right)^2 },
	\end{equation}
	where	$V=\hat{A}^2p_{{u}} \eta_{u}^{-1}\eta_{d}^{-1}/N_0$. It is evident from  \eqref{eq:snrALL} that the maximum SNR is obtained by taking the channel phases  as $\phi_k-\theta_{u_k}-\theta_{d_k}=0$, $\forall k \in ({-n,-n+1,\cdots,n-1,n})$ which maximizes the  exponential term to  unity  \cite{basar2019wireless}. 
	Now the modified maximum SNR is given as follows:
	\begin{equation}\label{eq:snrALLMaxIRS}
	\Gamma_{\rm IRS}=V{ \;\left(\sum_{k=1}^{N} d_{u_k}^{-\frac{\alpha(\theta_{u_k})}{2}} d_{d_k}^{-\frac{\alpha(\theta_{d_k})}{2}}\vert h_{u_k} \vert\vert h_{d_k} \vert  \right)^2 }.
	\end{equation}
	Given the limited size of UAV and the IRS, we assume that the distance between  {\bf S} and $k$-th IRS element is approximately the same as the distance between {\bf S} and UAV. Similarly, we assume that the distance between  {\bf D} and $k$-th IRS element is approximately the same as the distance between {\bf D} and UAV.
	That is,  $d_{u_k}\approx d_u$, $d_{d_k}\approx d_d$, $\theta_{u_k}\approx \theta_{u}$, and $\theta_{d_k}\approx \theta_{d}$. From this point onward, we will use $\alpha_u$  and $\alpha_d$ as $\alpha(\theta_{u_k})$, and $\alpha(\theta_{d_k})$, respectively, for brevity. 
	Subsequently, \eqref{eq:snrALLMaxIRS} simplifies as follows:
	\begin{align}\label{eq:snrALLMaxIRS2}
	\Gamma_{\rm IRS}\approx V	d_{u}^{-\alpha_u} d_{d}^{-\alpha_d} 	{ \;  (\sum_{k=1}^{N } \vert h_{u_k} \vert \vert h_{d_k} \vert  )^2 }.
	\end{align}
	Eq. {\eqref{eq:snrALLMaxIRS} and its approximation in \eqref{eq:snrALLMaxIRS2} are validated in Fig.~\ref{fig:Approxi_eq8and_eq9} for different simulation  parameters. }
	Note that, the IRS implicitly operates in full-duplex mode (with zero self-interference) and the incident signals on IRS reflect with minimal delay (typically less than the decoding delay experienced in DF relaying). 
	
	\subsubsection{Integrated UAV-IRS  Mode}
	Here, both the UAV and the IRS relay the signal transmitted from \textbf{S} and the receiver uses SC to extract the desired signal. The SNR at the receiver can be formulated as follows:
	\begin{align} \label{eq:SCMOdeSNR}
	\Gamma_{\rm INT}=\max{(\Gamma_{\rm UAV},\Gamma_{\rm IRS})}=
	\max\left\{	 \min \left( {p_{{u}} \kappa_u d_u^{-\alpha_u}\;X_{u}},p_{{d}}
	\kappa_d {   \; d_d^{-\alpha_d} \;X_{d}}  \right), 
	V		\frac{ \;  (\sum_{k=1}^{N } \vert h_{u_k} \vert \vert h_{d_k} \vert  )^2 }{d_{u}^{\alpha_u} d_{d}^{\alpha_d}}\right\}.
	\end{align}
	\normalsize

	\subsection{Energy Consumption  Model}
	\label{subsec:PowerModel}
	
	We consider that the UAV hovering time is equal to the time UAV can communicate and can be computed as $T_{hov}=\frac{E_B}{p_{\rm uav}}$, where $E_B$ is maximum UAV battery capacity and  $p_{\rm uav}$ is the power consumption of the UAV. 	The total power consumption of the considered system includes (i) the power consumed by the UAV for hovering and supporting IRS transmissions ($p_{\rm uav}$) and data transmission ($p_d$) in downlink and (ii) the hardware power consumption $(p_{bs})$ of the ground BS transmitter and receiver as well as $(p_{u})$.

	\subsubsection {UAV Power Consumption ($p_{\rm uav}$)}
	The total UAV power consumption is the sum of powers consumed by UAV in hovering $p_h$, circuit power consumption $p_c$ \cite{lu2017energy}, and the power consumed by UAV in the IRS hardware $p_{\rm IRS}$. 
	That is, the UAV power consumption can be given as 
	$
	p_{\rm uav}= p_c+p_{\rm IRS}+p_h.
	$
	where  $p_h=\frac{\delta}{8}\rho s A \xi^3r^3+(1+\kappa)\sqrt{\frac{(mg)^{3}}{{2\rho A}}},$ in which  $\rho$,  $A$, $\xi$, $r$,   $s$, $\delta$, and   $\kappa$ denote the air density (in kg/m$^3$), rotor disc area (in m$^2$), blade angular velocity (in rad/sec), rotor radius (in m),  rotor solidity,  profile drag coefficient, and incremental correction  factor of induced power,  respectively.
	
	Since we consider that the IRS is mounted on a UAV, IRS power consumption is a part of the total UAV power consumption. Note that, IRS is acting as a passive device and  does not need any transmission power.  However, its power consumption is due to the number of IRS elements and the phase resolution \cite{huang2019reconfigurable} and is  thus written as $p_{\rm IRS} ={N P_r(b)}$, where $P_r(b)$ is phase resolution power consumption. For instance, the power consumption of  finite phase resolution for 6 bits is $P_r(6)=$ 78mW and  for infinite phase resolution is $P_r(\infty)=$ 45dBm (Fig.~4 of \cite{huang2018energy}).  Therefore, an increase in the resolution and the number of IRS elements increases its  hardware power consumption as formulated in  \cite{huang2018energy,huang2019reconfigurable}.  
	
	\subsubsection {Terrestrial Circuit Power Consumption $(p_{bs})$} It is the hardware power consumption, i.e. the circuit power consumed by the source and destination ground BSs \cite{bousia2014energy} given as $p_{bs}$. 
	
	\subsubsection{Transmission Power Consumption} The transmission power consumption includes transmission power of the source BS in the uplink ($p_u$) and that of the UAV in the downlink ($p_d$). 
	
	
	Subsequently, we can define the total power consumption of each transmission mode as $P_{\rm  UAV}=p_u+p_d+C$, 	$P_{\rm IRS}=p_u+p_{\rm IRS}+C$, and $P_{\rm{INT}}=p_u+p_d+p_{\rm IRS}+C$, where $C=p_c+p_h+2p_{\rm{bs}}$.

	
	\section{Performance Characterization of Integrated UAV-IRS Relaying}
	\label{AnalysisSection}
	In this section, we characterize the  outage probability $O_m$, ergodic capacity $C_m$ and energy-efficiency ${\rm EE_m}$ for each of the modes (i.e. UAV-only, IRS-only, and integrated UAV-IRS modes) for the considered integrated UAV-IRS  relaying system. The subscript $m$ denotes the mode of operation.
	
	\subsection{UAV-only Mode of Relaying}
	Conditioned on the distances $d_u$ and $d_d$, the end-to-end SNR $\Gamma_{\rm UAV}$ can be given using \eqref{eq:SNRDFMode}. Subsequently, the SNR outage probability can be defined as follows:
	\begin{align}	\label{eq:Pout2Rice}
	\begin{split}
	{O}_{\rm UAV}=&\mathbb{P}(\Gamma_{\rm UAV}<\Gamma_0)=\mathbb{P}[\min(\gamma_u,\gamma_d)<\Gamma_0]
	=  1-\left(1-F_{\gamma_{u}}(\Gamma_0)\right) \left(1-F_{\gamma_{{d}}}(\Gamma_0)\right),	
	\end{split}
	\end{align}
	where $F_{\gamma_{u}}(\Gamma_0)$ and $F_{\gamma_{{d}}}(\Gamma_0)$ represent the CDFs of the SNR received on the channel from {\bf S} to UAV and UAV to {\bf D}, respectively, evaluated at the desired SNR threshold $\Gamma_0$. Using \eqref{eq:SNRup}, the $i$-th link SNR outage can be given  as follows:
	\begin{align}\label{1}
	\begin{split}
	F_{\gamma_{i}}(\Gamma_0)=&\mathbb{P}\left(\gamma_{i}\le \Gamma_0 \right)=\mathbb{P}\left( X_{i}\le \Gamma^{\prime}_{i} d_i^{-\alpha_i} \right)=F_{X_{i}}(\Gamma_i^\prime d_i^{-\alpha_i}), \qquad i \in\{u,d\}
	\end{split}
	\end{align}
	where $X_{i}=\vert h_{i}\vert^2$ represents non-central chi square distribution  and $\Gamma^{\prime}_{i}= \frac{ \Gamma_0}{\kappa_i p_{{i}}  }$.  Using the alternate exact expression for PDF in  \eqref{eq:LOSdistribution}, the CDF of $X_i$ can be given as follows \cite{nguyen2018two,bhatnagar2013capacity,nguyen2018energy}:
	\begin{equation} \label{eq:CDFSingle}	
	F_{X_i}(x_i) =  1- \sum_{\ell=0}^{\infty} \sum_{m=0}^{\ell} f_i(m,l) x_i^{m} e^{-b x_i},	
	\end{equation}
	where   $f_i(m,\ell)=e^{-K_i}\frac{K_i^{\ell} b_i^{m}}{\ell! m! }$, $b_i =\frac{K_i+1}{\Omega_i}$, $\Omega_i$ is the mean local power of the {Rician}   channel in the $i$-th link, and $K_i$ is the Rician factor.  Substituting $x_i=\Gamma_i^\prime d_i^{\alpha_i}$ in \eqref{eq:CDFSingle}, we obtain
	\begin{equation}
	F_{\gamma_{i}}(\Gamma_0)=F_{X_{i}}(\Gamma_i^\prime d_i^{\alpha_i})=
	1-\sum_{\ell=0}^{\infty} \sum_{m=0}^{\ell} f_i(m,l) \left({\Gamma_i^\prime d_i^{\alpha_i}}\right)^{m} \exp \left(- b_i\Gamma_i^\prime d_i^{\alpha_i} \right), \qquad i \in(u,d).
	\label{eq:aaaa}
	\end{equation}
	By using \eqref{eq:aaaa} for  $i=u$ and  $i=d$ in \eqref{eq:Pout2Rice}, the end-to-end SNR outage $O_{\rm UAV} $ is given  as
	\begin{align}	\label{eq:OutDF}
	O_{\rm UAV} =&  1- \sum_{\ell=0}^{\infty} \sum_{m=0}^{\ell} f_u(m,l) \left({\Gamma_u^\prime d_u^{\alpha_u}}\right)^{m} \exp \left(- b_u\Gamma_u^\prime d_u^{\alpha_u} \right)
	\sum_{\ell=0}^{\infty} \sum_{m=0}^{\ell} f_d(m,l) \left({\Gamma_d^\prime d_d^{\alpha_d}}\right)^{m} \exp \left(- b_d\Gamma_d^\prime d_d^{\alpha_d} \right).		\end{align}
	\begin{corollary}
		\label{Corollary1} 
		In scenarios where  NLoS components are dominant (i.e. for $K_u=0$ and $K_d=0$) than LoS components, the Rician distribution follows Rayleigh distribution. As such, the CDF in \eqref{eq:aaaa} can be expressed as $
		F_{X_{u}}(\Gamma_u^\prime d_u^{-\alpha_u}) = 1-e^{-{\Gamma^\prime_u d_u^{-\alpha_u}}/{\Omega_u} },
		F_{X_{d}}(\Gamma_d^\prime d_d^{-\alpha_d}) =  1- e^{-\frac{\Gamma_d^\prime}{\Omega_d}{{ d_d^{\alpha_d}}}}.$ The end-to-end SNR outage for UAV-only mode of relaying can be simplified as follows:
		\begin{align}	\label{eq:Event4}	O_{\rm UAV} (d_d)=  1-  e^{-\frac{\Gamma^\prime_u}{\Omega_u} {{ d_u^{-\alpha_u}}} -\frac{\Gamma_d^\prime}{\Omega_d}{{ d_d^{-\alpha_d}}}  }.
		\end{align}
	\end{corollary}
	
	\subsection{Outage Probability for IRS-only Mode of Relaying}
	\label{subsec:Outend2endMode2}
	
	Using \eqref{eq:snrALLMaxIRS2}, the end-to-end SNR outage in the IRS-only mode of relaying can be given  as:
	\begin{align}\label{1IRS}
	\begin{split}
	O_{\rm IRS}=&\mathbb{P}\left(\Gamma_{\rm IRS}\le \Gamma_0 \right)=\mathbb{P}\left(V	d_{u}^{-\alpha_u} d_{d}^{-\alpha_d} 	{ \; Z^2 }\le \Gamma_0 \right)=\mathbb{P}\left( Z^2 \le t \Gamma_0 \right),
	\end{split}
	\end{align}
	where $Z=\sum_{k=1}^{N}\vert h_{u_k} \vert\vert h_{d_k} \vert $ and $t=\frac{	d_{u}^{\alpha_u} d_{d}^{\alpha_d}}{V }$. The  SNR outage probability can then be derived as follows. 
	\begin{prop} 
		\label{sec:IRSOutage}
		The outage probability of IRS-only mode can be given as follows:
		\begin{align}\label{33}
		O_{\rm IRS}=\frac{1}{2} \left(\erf   \left(\frac{\sqrt{t\Gamma_0}-\sqrt{\lambda}}{\sqrt{2}} \right)+  \erf   \left(\frac{\sqrt{t\Gamma_0}+\sqrt{\lambda}}{\sqrt{2}} \right)\right),
		\end{align}
		where $t=\frac{	d_{u}^{\alpha_u} d_{d}^{\alpha_d}}{V }$, $\lambda=\frac{1}{2}\frac{\mu_Z^2}{\sigma_Z^2}$, $\mu_z=(N+1) \mathbb{E}[\vert h_{u_k}\vert\vert h_{d_k} \vert]$, and $\sigma^2_Z=(N+1) var(\vert h_{u_k}\vert\vert h_{d_k} \vert)$.
	\end{prop}	
	\begin{proof} 
		In general, a meta-surface is made up of a large number of reflecting elements, i.e. $N\gg 1$. Therefore, we apply  central limit theorem (CLT) on $Z=\sum_{k=1}^{N}\vert h_{u_k} \vert\vert h_{d_k} \vert$, where $\vert h_{u_k}\vert$ and  $\vert h_{d_k} \vert$ are identically and independently distributed (i.i.d) random variables with the mean and variance $\mathbb{E}[ \vert h_{i_k}\vert] $ and $ var( \vert h_{i_k}\vert) $, respectively,  for $i\in (u,d)$. Subsequently, the distribution of $Z$ will converge to the Gaussian distribution with mean and variance, respectively, given by $$\mu_z=(N+1) \mathbb{E}[\vert h_{u_k}\vert\vert h_{d_k} \vert]\qquad \text{and} \qquad \sigma^2_Z=(N+1) var(\vert h_{u_k}\vert\vert h_{d_k} \vert).$$ 
		Note that $\vert h_{u_k}\vert$ and $\vert h_{d_k} \vert$ are independent, but  may not be identically distributed  Rician  variables. Therefore, we consider that the product $\vert h_{u_k}\vert\vert h_{d_k} \vert$ is the product of two independent but non-identical Rician random variables. The product follows the double-Rician distribution \cite{talha2007statistical} with the mean and variance given as: 
		$$\mathbb{E}[\vert h_{u_k}\vert\vert h_{d_k} \vert]=  \sigma \frac{\pi}{2}
		{}_1F_1\left( \frac{-1}{2}; 1; \frac{-\mu_u^2}{2\Omega_u} \right) 	{}_1F_1\left( \frac{-1}{2}; 1; \frac{-\mu_d^2}{2\Omega_d} \right),	$$ and
		$$var(\vert h_{u_k}\vert\vert h_{d_k} \vert)=2^2 \sigma^2 \left( 1+\frac{\mu^2_u}{2\Omega_u}\right)\left( 1+\frac{\mu^2_d}{2\Omega_d}\right)-\left( \sigma \frac{\pi}{2} \right)^2\left[ 	{}_1F_1\left( \frac{-1}{2}; 1; \frac{-\mu_u^2}{2\Omega_u} \right) 	{}_1F_1\left( \frac{-1}{2}; 1; \frac{-\mu_d^2}{2\Omega_d} \right) \right]^2,$$ 
		where $\sigma^2=\Omega_u\Omega_d$ and ${}_1F_1\left( . \right) $ is the Confluent Hypergeometric function. Now taking  $X=Z^2$, the distribution of $X$ follows the non-central chi square distribution with unity degree of freedom  and non-centrality parameter $\lambda=\frac{1}{2}\frac{\mu_Z^2}{\sigma_Z^2}$.  Subsequently, the probability density function (PDF) of $X$ is given as:
		\begin{align}\label{eq:RicianUpeqold}
		f_{X}(x)& = 
		\frac{1}{2}\left(\frac{x}{\lambda}\right)^{-1/4}  e^{-\frac{\lambda+x}{2}} I_{-1/2} \left({ \sqrt{\lambda x }}\right),
		\end{align}	
		where
		$I_{\beta}$ is the modified Bessel function of first kind of order $\beta$. {Fig.~\ref{fig:CentralLimitTheoremVarificationVsN} shows that the PDF of $X$ obtained from simulations  converges to non-central chi square variable for $N\ge 20$, as is   implied by CLT.}
		\begin{align}\label{OutIRS1}
		\begin{split}
		O_{\rm IRS}=&\mathbb{P}\left( Z^2 \le t \Gamma_0 \right)=\mathbb{P}\left( -\sqrt{ t \Gamma_0}\le X \le \sqrt{ t \Gamma_0}\right)=\int_{x=0}^{t\Gamma_0}	\frac{1}{2}\left(\frac{X}{\lambda}\right)^{-1/4}  e^{-\frac{\lambda+X}{2}} I_{-1/2} \left({ \sqrt{\lambda X }}\right) dX.
		\end{split}
		\end{align}
	\end{proof}
	For	$\mu_u=\mu_d=0$, the mean and variance of the double Rician variable can be simplified as follows.
	\begin{corollary}
		\label{Corollary2} For	$\mu_u=\mu_d=0$, the  double Rician variable converts to double Rayleigh variable. Thus,
		the mean and variance of the product $\vert h_{u_k}\vert \vert h_{d_k} \vert$  can be simplified as  $\mathbb{E}[\vert h_{u_k}\vert\vert h_{d_k} \vert]=\sigma \frac{\pi}{2}$ and
		$var(\vert h_{u_k}\vert\vert h_{d_k} \vert)=2^2 \sigma^2 (1-\pi^2/16)$ with $\sigma^2=\Omega_u\Omega_d$ \cite{salo2006distribution}, respectively.  After applying central limit theorem for Rayleigh fading, we obtain $\mu_z=(N+1)\sigma \frac{\pi}{2}$, $\sigma^2_Z=(N+1) 2^2 \sigma^2 (1-\pi^2/16)$.
	\end{corollary}
	\begin{figure*}[t]
		\begin{minipage}{0.35\textwidth}
			\includegraphics[width=8cm,height=6cm]{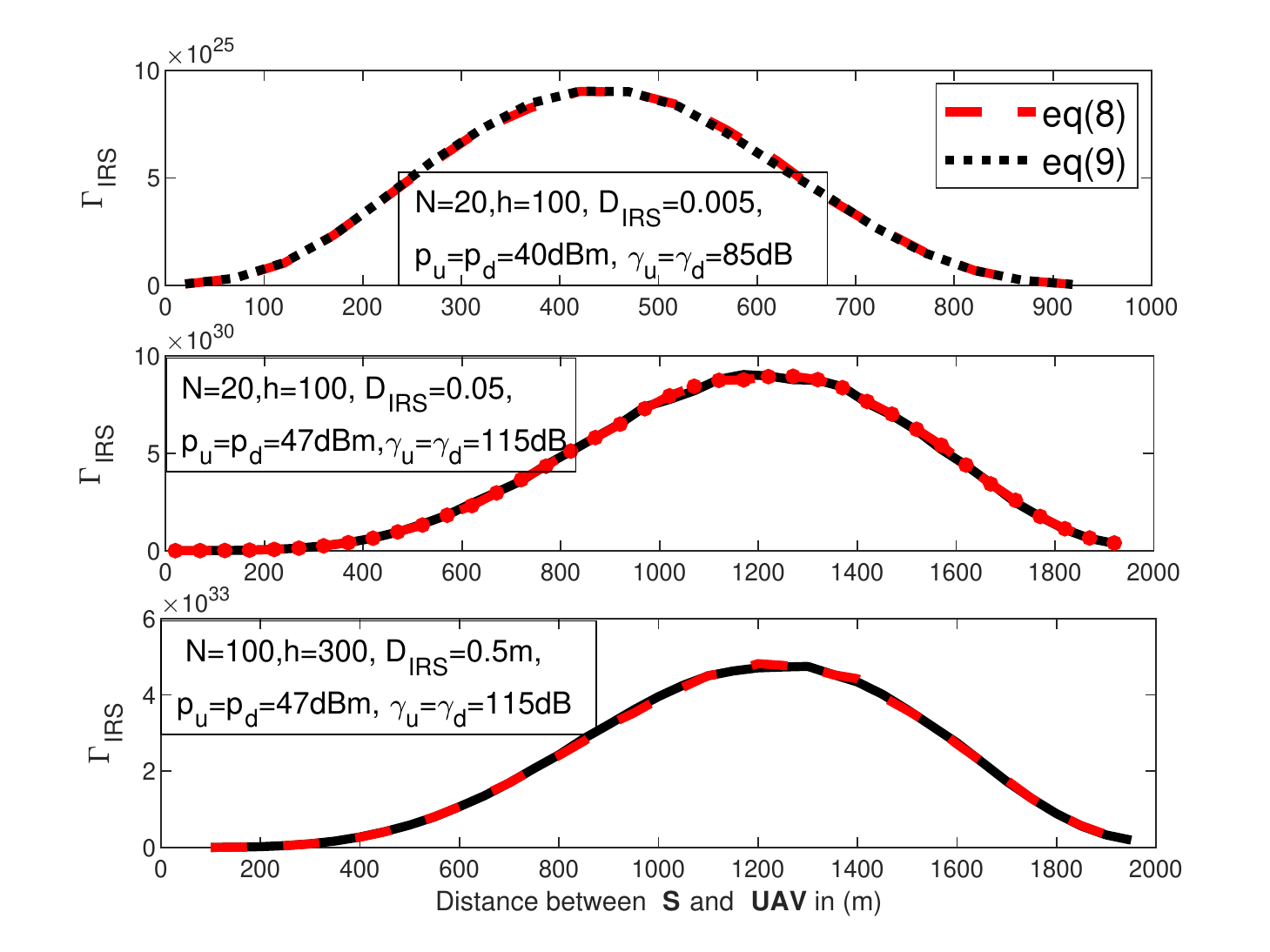}
			\caption{Comparison of \eqref{eq:snrALLMaxIRS} and its approximation in \eqref{eq:snrALLMaxIRS2} for different parameters.}
			\label{fig:Approxi_eq8and_eq9}	
		\end{minipage}\hfill
		\begin{minipage}{0.55\textwidth}
			\includegraphics[width=10cm,height=6cm]{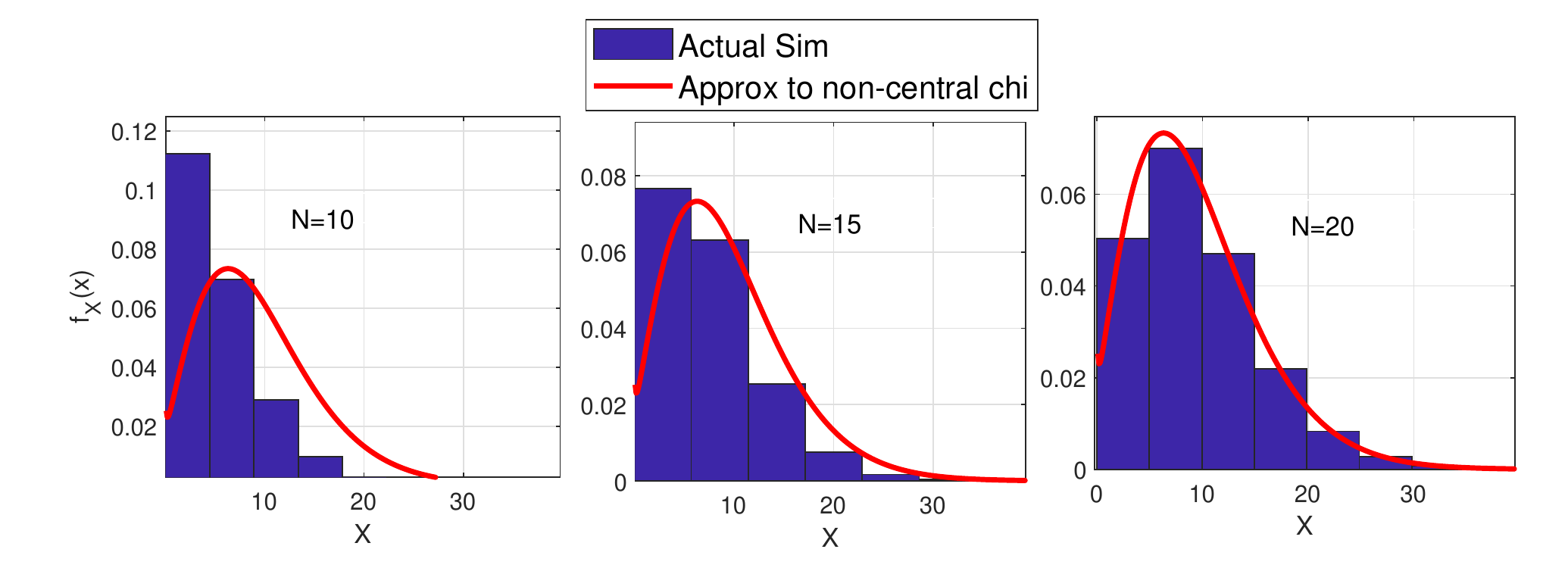}
			\caption{Comparison of the exact PDF of $X$ through simulations and the PDF  of $X$ when $N \rightarrow \infty$ that follows non-central chi-square distribution. }
			\label{fig:CentralLimitTheoremVarificationVsN}	
		\end{minipage}
	\end{figure*}
	\subsection{Outage Probability of Integrated UAV-IRS Mode of Relaying}
	\label{subsec:Outend2endMode3}
	In integrated UAV-IRS mode of relaying, we assume that the receiver applies SC and selects the mode of operation associated to the maximum SNR. This implies additional degree of freedom, however, at the expense of increased resource consumption,  since both the  UAV and  the IRS are actively transmitting to $\bf{D}$. The outage probability of this mode can thus be derived using \eqref{eq:SCMOdeSNR} as follows: 
	\begin{align} \label{eq:SelCombining}
	O_{\rm INT}=&
	\Pr\left(\Gamma_{\rm UAV} \le \Gamma_0\right)
	\Pr\left(\Gamma_{\rm IRS} \le \Gamma_0\right)=O_{\rm UAV}O_{\rm IRS},
	\end{align}
	where $O_{\rm UAV}$ and $O_{\rm IRS}$ are given in Section~III.A and III.B, respectively.
	\subsection{Ergodic Capacity $C_m$ and Energy Efficiency ${\rm{EE}}_m$ for Mode $m$ }
	Given mode $m$, the exact end-to-end ergodic capacity at the receiver can be derived  as follows \cite{tabassum2018coverage}: 
	\begin{align}	\label{eq:RateMode1}
	C_{m}=&\mathbb{E} [B\mathrm{log}_2(1+\Gamma_{m} )]
	=\frac{1}{\rm ln(2)}\int_{t=0}^{\infty} \frac{\Pr(\Gamma_m>t)}{1+t}dt = \frac{B}{\rm ln(2)}\int_{\Gamma_0=0}^{\infty} \frac{1-O_{m}}{1+\Gamma_0}\; d\Gamma_0,
	\end{align}
	where $\Pr\left(\Gamma_m>\Gamma_0\right)=1-O_{m}$ and $O_m$ is derived in \eqref{eq:OutDF}, \eqref{33}, and \eqref{eq:SelCombining} for UAV-only, IRS-only and integrated UAV-IRS modes, respectively.
	Along the similar lines, using the definition of  energy efficiency EE$_m$ of each mode $m$ (which is defined as the ratio of ergodic capacity to the corresponding power consumption $P_m$), we can derive the exact end-to-end energy-efficiency as follows: 
	\begin{align}\label{eq:EE_m}
	\begin{split}
	{\rm EE}_{m}= \frac{B}{\rm ln(2)P_{m}}\int_{\Gamma_0=0}^{\infty} \frac{1-O_{m}}{1+\Gamma_0}\; d\Gamma_0.
	\end{split}
	\end{align}

	\section{Approximate Performance Characterizations for UAV-IRS Relaying}
	\label{Sec:EEEachMode}
	In this section, we first derive a bound on the ergodic capacity $C_m$ and energy efficiency ${\rm EE_m}$ for each mode of relaying (i.e. UAV-only, IRS-only, and UAV-IRS modes). 
	
	
	\noindent
	{\em UAV-only Mode of Relaying}:
	Applying Jensen's Inequality to the ergodic capacity expression, an upper bound on the ergodic capacity  (in bps)  can be derived as follows:
	\begin{align}	\label{eq:RateMode2jensen}
	\mathbb{E} \left[\mathrm{log}_2\left( 1+ {\rm SNR_m} \right) \right] \leq \mathrm{log}_2\left( 1+ \mathbb{E} \left[{\rm SNR_m} \right]\right).
	\end{align}
	Subsequently, we derive tractable expressions of the ergodic capacity and energy-efficiency in UAV-only mode as shown in the following Proposition.
	\begin{prop} 
		The ergodic capacity $C_{\rm UAV}$ and energy-efficiency $EE_{\rm UAV}$ expressions in UAV-only mode can be given, respectively, as follows: 
		\begin{equation}
		\label{eq:C_UAVFinal}
		\begin{split}
		C_{\rm UAV}
		{\leq}& 	B\mathrm{log}_2\left( 1+\min \left(   {p_{{u}} \kappa_u \; d_u^{-\alpha_u}  \; \;\Omega_{u} },{p_{{d}} \kappa_d \; d_d^{-\alpha_d}  \; \Omega_{d} }    \right)\right), 
		\end{split}
		\end{equation}
		\begin{align}\label{eq:EE_DFApprox} 
		\begin{split} 
		\rm{EE}_{UAV}\approx& \frac{B\mathrm{log}_2\left( 1+\min \left(   {p_{{u}} \kappa_u \; d_u^{-\alpha_u}  \; \;\Omega_{u} },{p_{{d}} \kappa_d \; d_d^{-\alpha_d}  \; \Omega_{d} }    \right)\right) }{{p_u+p_d +C}}.
		\end{split} 
		\end{align} 
	\end{prop}
	\begin {proof}
	The ergodic capacity  $C_{\rm UAV}$ in  \eqref{eq:RateMode1} can be bounded as follows:
	\begin{equation}
	\label{eq:RicianApproxDn4Node1}
	\begin{split}
	C_{\rm UAV}
	\stackrel{(a)}{\leq}& B\mathrm{log}_2\left( 1+\mathbb{E} \left[ \min \left( {p_{{u}} \kappa_u\; d_u^{-\alpha_u}  \; \;X_{u}},{p_{{d}} \kappa_d \; d_d^{-\alpha_d}  \;X_{d}}  \right) \right] \right)
	\\
	\stackrel{(b)}{\approx}&B\mathrm{log}_2\left( 1+\min \left( \mathbb{E} \left[  {p_{{u}} \kappa_u\; d_u^{-\alpha_u}  \; \;X_{u}},{p_{{d}} \kappa_d \; d_d^{-\alpha_d}  \;X_{d}}   \right] \right)\right)\\ 
	\stackrel{(c)}{=} &
	B\mathrm{log}_2\left( 1+\min \left(   {p_{{u}} \kappa_u \; d_u^{-\alpha_u}  \; \;\mathbb{E} \left[X_{u}\right]},{p_{{d}} \kappa_d \; d_d^{-\alpha_d}  \; \mathbb{E}\left[X_{d}\right]}    \right)\right),
	\end{split}
	\end{equation}
	where $X_{u}$ and $X_{d}$ follow  non-central chi square distribution. 
	Note that (a) is obtained by using Jensen's inequality \cite{jiang2010dynamic}, (b) is obtained by interchanging $\min(.)$ and $\mathbb{E}(.)$ (validated in Fig.~4),  (c) follows from non-central chi-square distribution with mean $\Omega_u$ and $\Omega_d$, respectively, and results in \eqref{eq:C_UAVFinal}.
	Finally, using ergodic capacity in  \eqref{eq:C_UAVFinal}, we obtain $\rm EE_{UAV}$ in \eqref{eq:EE_DFApprox}.
\end{proof}
Fig.~\ref{fig:CapacityApproxUAVMode} validates the accuracy of our proposed bounds in \eqref{eq:RicianApproxDn4Node1}(step a) using Jensen's inequality and \eqref{eq:RicianApproxDn4Node1}(step b) using interchange of min($\cdot$) and the expectation operator $\mathbb{E}[\cdot]$ with exact Monte-Carlo simulations. To further justify the approximation in (b), we calculate the expectation of the minimum of two random variables, i.e. $\mathbb{E}[\Gamma_{\rm UAV}]=\mathbb{E}[\min \left( {p_{{u}} \kappa_u\; d_u^{-\alpha_u}  \; \;X_{u}},{p_{{d}} \kappa_d \; d_d^{-\alpha_d}  \;X_{d}}  \right)]$ in an exact form. That is, we first determine the PDF of $ \Gamma_{\rm UAV}$, by taking the derivative of the CDF of $\Gamma_{\rm UAV}$. The CDF of $\Gamma_{\rm UAV}$ can be derived
using \eqref{eq:OutDF}, by substituting $\Gamma_u^\prime$ and $\Gamma_d^\prime$ and replacing $\Gamma_0$ with $z$. 
Finally, we calculate $\mathbb{E}[\Gamma_{\rm UAV}]=\int_{z=0}^{\infty} z f_{\rm \Gamma_{UAV}}(z) dz $, under the condition that $m_u+m_d$ and $m_u+m_d>0$ and $\left( \frac{b_u d_u^{\alpha_u} }{\kappa_u p_{{u}}  } + \frac{ b_d d_d^{\alpha_d}}{\kappa_d p_{{d}}  } \right)\ge 0 $.


\begin{corollary}
	\label{Corollary3} 
	When  NLoS components are dominant, the Rician distribution follows Rayleigh distribution, i.e. $\Omega_u=1$ and $\Omega_d=1$. The end-to-end $\rm{EE}_{UAV}$ for UAV-only mode in \eqref{eq:EE_DFApprox} can be simplified as follows:
	\begin{align}\label{eq:EE_DFExactRayleigh} 
	\begin{split} 
	\rm{EE}_{UAV}\approx& \rm \frac{B\mathrm{log}_2\left( 1+\min \left(   {p_{{u}} \kappa_u \; d_u^{-\alpha_u}  },{p_{{d}} \kappa_d \; d_d^{-\alpha_d} }    \right)\right) }{{p_u+p_d +C}}.
	\end{split} 
	\end{align}
\end{corollary}

\noindent
{\em IRS-only Mode of Relaying}:
For IRS-only mode, the  ergodic capacity and EE expressions are derived in the following.
\begin{prop}The  ergodic capacity expression can be obtained for IRS-only mode as follows:
	\begin{align}	\label{eq:RateMode2Approx}
	\begin{split}
	C_{\rm IRS}\stackrel{(a)}{\leq} B \mathrm{log}_2\left(1+	\mathbb{E}\left[V	d_{u}^{-\alpha_u} d_{d}^{-\alpha_d} 	{ \; X }\right ]\right) \stackrel{(b)}{=} B\mathrm{log}_2\left(1+V	d_{u}^{-\alpha_u} d_{d}^{-\alpha_d} 	{ \;  (	v+\lambda) } \right),
	\end{split}
	\end{align}
	where (a) is obtained using Jensen's inequality and (b) is obtained using $\mathbb{E}[X]=\nu+\lambda$ \cite{DgitalCommProakis}. Using \eqref{eq:RateMode2Approx}(step b), we bound ${\rm EE}_{\rm IRS}$  as follows:
	\begin{align}\label{eq:EEIRSApprox} 
	\begin{split} 
	{\rm EE}_{\rm IRS}\approx \rm \frac{ B\mathrm{log}_2\left( 1+ V	d_{u}^{-\alpha_u} d_{d}^{-\alpha_d} 	{ \; (\nu+\lambda)} \right)}{P_{\rm IRS} }.
	\end{split} 
	\end{align} 
\end{prop}

\noindent
{\em Integrated UAV-IRS Mode of Relaying}:
For integrated UAV-IRS mode ($m=$ INT), the  ergodic capacity in \eqref{eq:RateMode1} can be bounded as:
\begin{align}	\label{eq:RateMode2}
C_{\rm INT}=& B\mathbb{E} \left[\mathrm{log}_2\left( 1+ \max\left(  V	d_{u}^{-\alpha_u} d_{d}^{-\alpha_d} 	{ \; X }  ,\min \left( {p_{{u}}  \kappa_u \; d_u^{-\alpha_u}  \;X_{u}},{p_d \kappa_{d}  \; d_d^{-\alpha_d}  \;X_{d}}  \right) \right) \right)\right],
\end{align}
where $X$, $X_{u}$, and $X_{d}$ follow  non-central chi square distribution representing end-to-end channel fading power in IRS transmission, channel fading power from {\bf S} to UAV and  UAV to {\bf D}, respectively. 	After applying SC, the ergodic capacity \eqref{eq:RateMode2} can be approximated as follows:
\begin{equation}
\label{eq:RAteSCApprox}
\begin{split}
C_{\rm INT}
\stackrel{(a)}{\leq}& B\mathrm{log}_2\left( 1+\mathbb{E} \left[\max\left(  V	d_{u}^{-\alpha_u} d_{d}^{-\alpha_d} X,\min \left( {p_{{u}} \kappa_{u} \; d_u^{-\alpha_u}  \;X_{u}},{p_d  \kappa_{d} \; d_d^{-\alpha_d}  \;X_{d}}  \right) \right) \right] \right)
\\
\stackrel{(b)}{\approx}& B\mathrm{log}_2\left( 1+\max\left(  V	d_{u}^{-\alpha_u} d_{d}^{-\alpha_d} 	{ \; \mathbb{E} \left[X \right]}  ,\mathbb{E} \left[\min \left( {\hat{A}p_{{u}} \eta_{u}^{-1} \; d_u^{-\alpha_u}  \;X_{u}},{\hat{A}p_d \eta_{d}^{-1} \; d_d^{-\alpha_d}  \;X_{d}}  \right)\right] \right) \right)\\
\stackrel{(c)}{\approx} &B\mathrm{log}_2\left( 1+\max\left(  V	d_{u}^{-\alpha_u} d_{d}^{-\alpha_d} 	{ \; \mathbb{E} \left[X \right]}  ,\min \left( {p_{{u}}  \kappa_{u} \; d_u^{-\alpha_u}  \;\mathbb{E} \left[X_{u}\right]},{p_d  \kappa_{d} \; d_d^{-\alpha_d}  \;\mathbb{E} \left[X_{d}\right]}  \right) \right) \right)\\
\stackrel{(d)}{=} &B\mathrm{log}_2\left( 1+\max\left(  V	d_{u}^{-\alpha_u} d_{d}^{-\alpha_d} 	{ \;( \nu+\lambda})  , \min \left(  {p_{{u}}  \kappa_{u} \; d_u^{-\alpha_u}  \; {\Omega_{u} } },{p_d  \kappa_{d} \; d_d^{-\alpha_d}  \;{\Omega_{d}}} \right) \right) \right),
\end{split}
\end{equation}	\normalsize
where (a) is obtained using Jensen's inequality, (b) and (c) are obtained by interchanging $\max(.)$ and $\min(.)$  operators with the $\mathbb{E}(.)$ operator, respectively, and (d) is obtained by substituting the mean of  $X$, $X_u$ and $X_d$ with 
$\nu+\lambda$, $\Omega_u$ and $\Omega_d$, respectively. Finally, using (d) we approximate ${\rm  EE}_{\rm INT}$  as follows:
\begin{align}\label{eq:EESCApprox} 
\begin{split} 
{\rm EE}_{\rm INT}\approx\rm \frac{ B\mathrm{log}_2\left( 1+\max\left(  V	d_{u}^{-\alpha_u} d_{d}^{-\alpha_d} { \; (\nu+\lambda)}  , \min\left(  {p_{{u}}  \kappa_{u} \; d_u^{-\alpha_u}  \; {\Omega_{u} } },{p_d  \kappa_{d} \; d_d^{-\alpha_d}  \;{\Omega_{d}}} \right) \right) \right)}{P_{\rm INT} }.
\end{split} 
\end{align}

\section{Optimization of UAV-IRS Relaying}
\label{Sec:4}
In this section, we consider two optimization problems for maximizing the network energy efficiency and minimizing the network power consumption subject to rate constraints, considering the UAV-only mode and the IRS-only mode of relaying. For the IRS-only mode, we optimize the number of active IRS elements $N$ and height of the IRS surface (i.e. UAV height). For the UAV-only mode, we optimize the UAV height.   

\begin{figure*}[t]
	\begin{minipage}{0.32\textwidth}
		\includegraphics[scale=0.42]{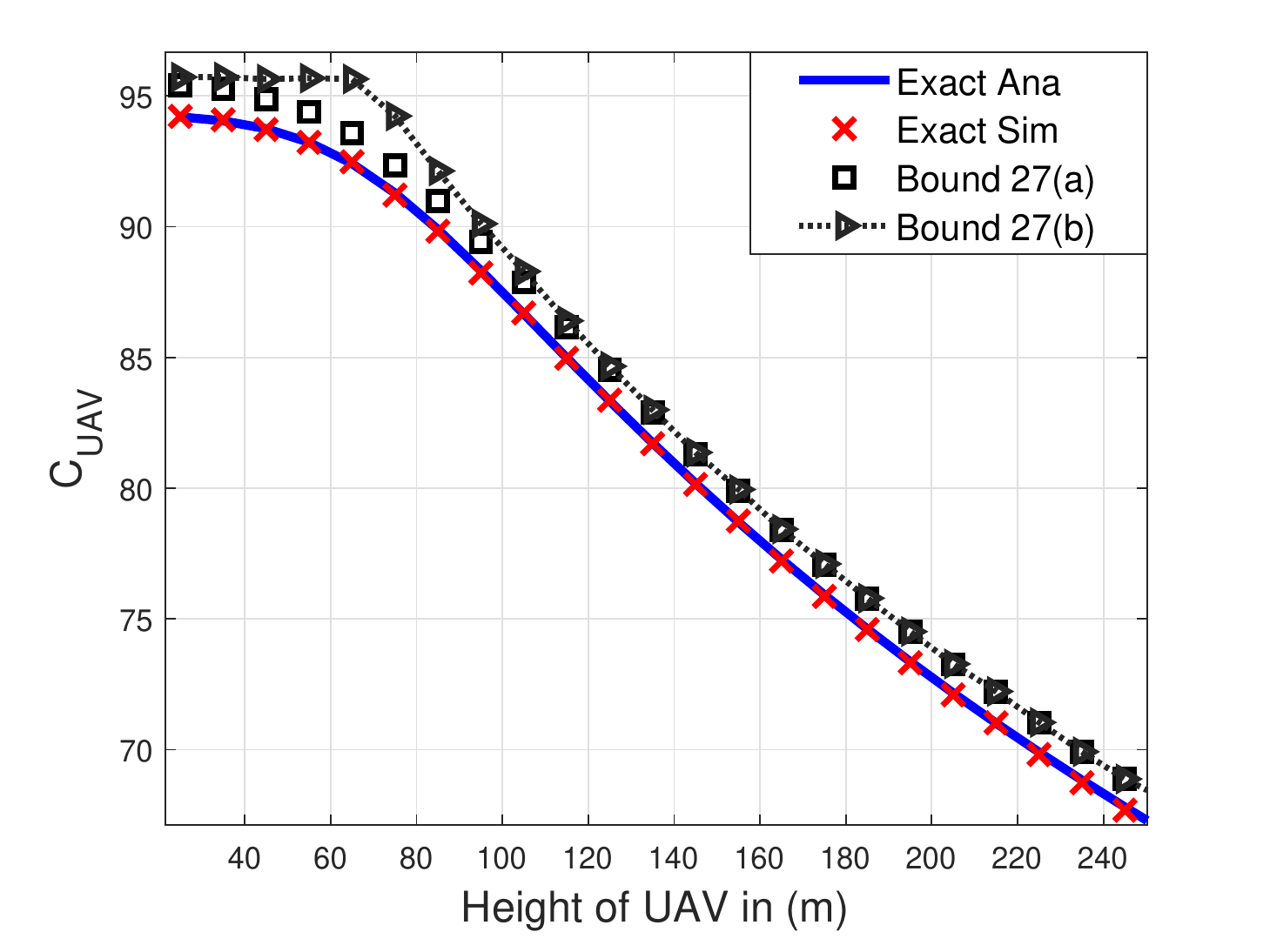}
		\caption{The comparison of exact ergodic capacity (via analysis and  simulations) to the UAV-only bounds provided in Eq. 27(step a) and  Eq. 27(step b). Approximation in Eq. 27(step b) is validated by analytically solving  Eq. 27(step a).}
		\label{fig:CapacityApproxUAVMode}	\end{minipage}\hfill
	\begin{minipage}{0.32\textwidth}
		\includegraphics[scale=0.42]{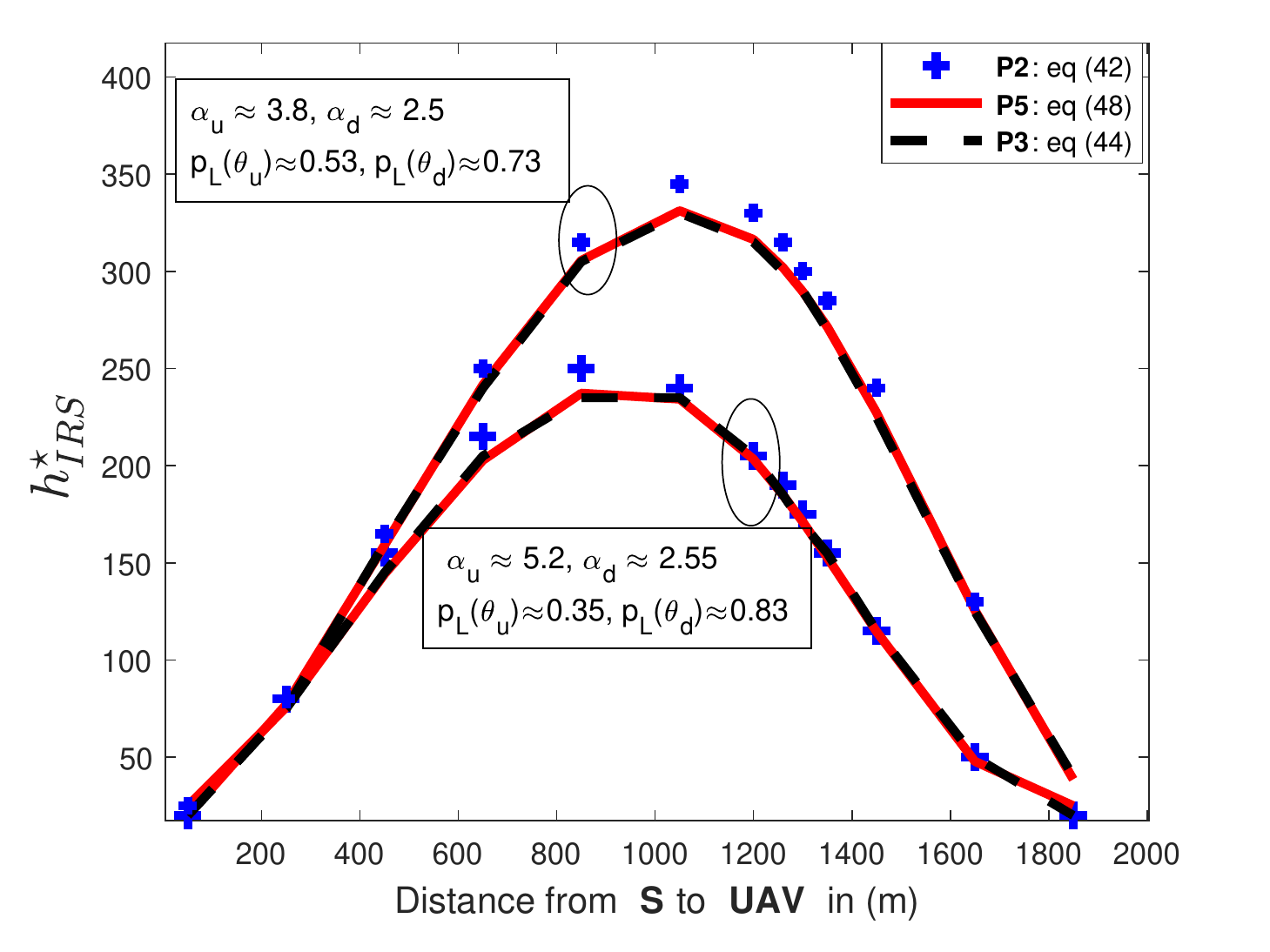}
		\caption{The validation of solution obtained from \eqref{eq:OPtProb1Var1in_d}, \eqref{eq:Concave_Numeratr} and \eqref{eq:2TermConcaveqUadratic} for IRS-only mode for 
			$E_b/N_0=130{ \rm dB}$, $R_{\rm SI}~=~38{\rm dB}$, $N=30$, $p_u=p_d=50$dBm for different environment parameters, where ($E_b$:~per symbol energy).}
		\label{fig:IRS_only_Mode_Approx_Comparison}	
	\end{minipage}\hfill
	\begin{minipage}{0.32\textwidth}
		\includegraphics[scale=0.42]{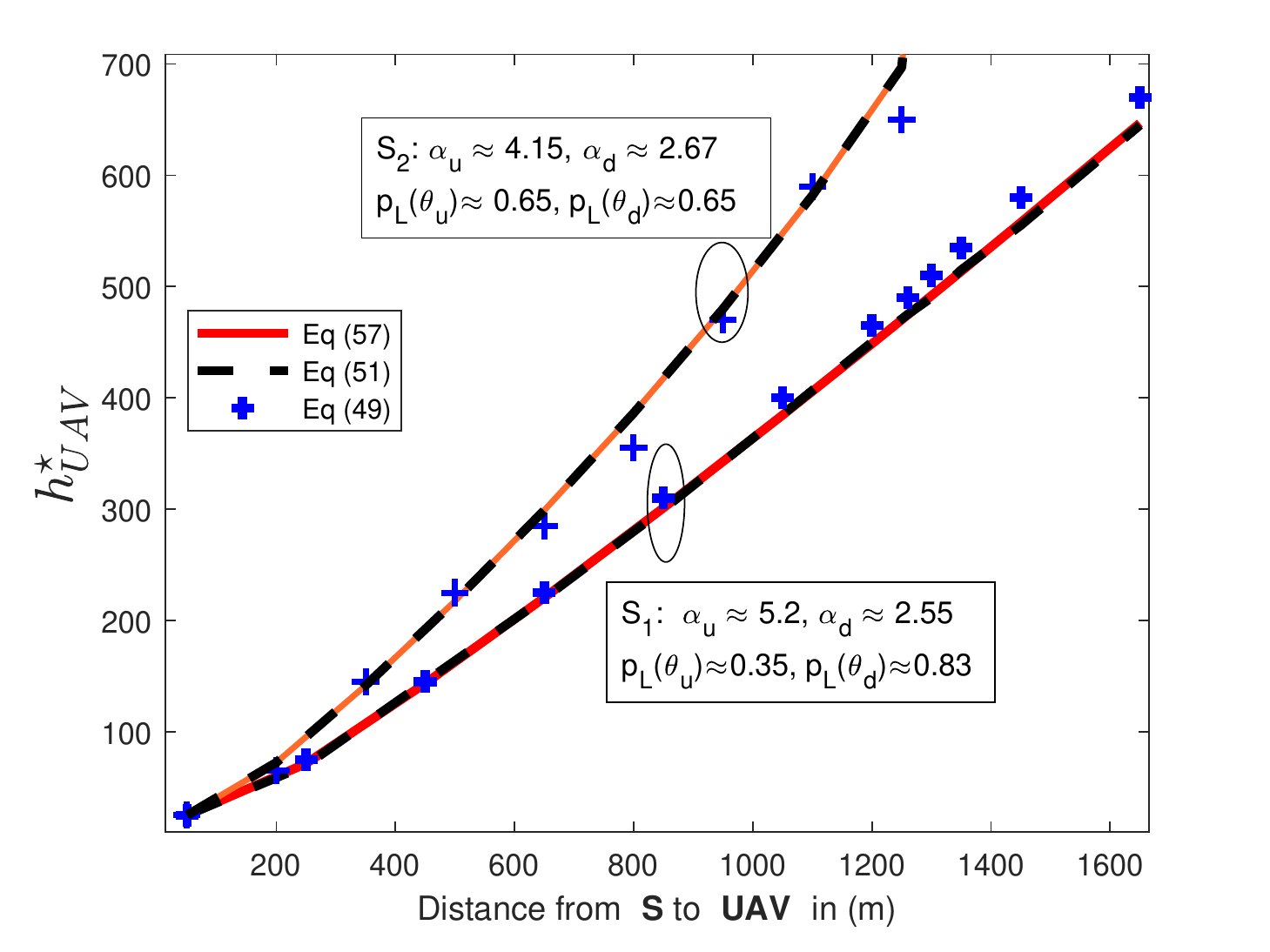}
		\caption{The validation of solution obtained from \eqref{eq:DFEnergyEE}, \eqref{eq:DFEnergyEE2} and \eqref{eq:DFEnergyEE5} for UAV-only mode for $E_b/N_0~=~135{\rm dB}$, $R_{\rm SI}~=~50{\rm dB}$, $ p_u=p_d=45$dBm for different environment parameters. For given parameters $I_i<10^{10}$ $\forall i\in(u,d)$ assures the concavity of $O_i(h)$. }
		\label{fig:UAV_only_Mode_Approx_Comparison}	
	\end{minipage}
\end{figure*}
\subsection{IRS-only Mode: Optimizing the Number of IRS Elements}
\subsubsection{EE Maximization} Using \eqref{eq:EEIRSApprox} where $\lambda$ is a function of $N$, i.e.
$\lambda=(N+1)\lambda^\prime$, where $\lambda^\prime= \frac{1}{2}\frac{\left( \mathbb{E}[\vert h_{u_k}\vert\vert h_{d_k} \vert]\right)^2}{ var(\vert h_{u_k}\vert\vert h_{d_k} \vert)}$ can be taken from \eqref{33}, the EE maximization problem can be formulated as follows:
\begin{align}
\begin{split}
{\bf P1}:	&\max_{N}\; 	 {\rm EE_{IRS}}=\rm \frac{B\mathrm{log}_2\left(1+V	d_{u}^{-\alpha_u} d_{d}^{-\alpha_d} 	({ \;  	v+\lambda }) \right)}{{p_u +NP_{r}(b)+C}}\\
{\rm s.t.}\; & {\bf C1}:   N_{\min} \leq N \le N_{\max}, \;\;\;\; 
\end{split}
\label{eq:OPtProb1Var1in_dN}
\end{align}
{where $N_{\rm max}$ is the maximum number of IRS elements that can be calculated as a ratio of the size of UAV to the size of one IRS element, $N_{\min}$ is the minimum number of IRS elements that can be deployed at a surface in practical settings and for which the objective function is accurate  [refer to Fig.~\ref{fig:CentralLimitTheoremVarificationVsN}].}
Since $\lambda$ is directly proportional to $N$, we reformulate the problem {\bf P1} as follows:
\begin{align}
\begin{split}
{\bf P2}:	&\max_{\lambda}\; 	\rm \frac{B\mathrm{log}_2\left(1+V	d_{u}^{-\alpha_u} d_{d}^{-\alpha_d} 	({ \;  	v+\lambda }) \right)}{ {p_u +(\frac{\lambda-\lambda^\prime}{\lambda^\prime})P_{r}(b)+C}}\\
{\rm s.t.}\; & {\bf C1}:  (N_{\min}+1)\lambda^\prime \le \lambda \le (N_{\max}+1)\lambda^\prime.
\end{split}
\label{eq:OPtProb1Var1in_dlambda}
\end{align}
{The problem P2 is non-convex in general; however, it is in the form of ratio of concave and convex function w.r.t variable $\lambda$. Fortunately, due to the structure of the problem, the global optimal solution can be obtained by applying quadratic transform proposed in \cite{shen2018fractional}. The quadratic transform converts the ratio of concave and convex function to the convex form by introducing an auxiliary variable $y$. Thus, we optimize
	the primal variable $\lambda$ and the auxiliary variable $y_j$ at each iteration $j$.

	The iterative algorithm is guaranteed to  converge to the global optimal solution for the single ratio objective function in {\bf P1}.} As such, using Quadratic Transform, the problem {\bf P2} can be reformulated as:
\begin{align}
\begin{split}
{\bf P3}:	&\max_{\lambda,y}\; 	Q(\lambda)=2y{{\sqrt{B\mathrm{log}_2\left(1+V	d_{u}^{-\alpha_u} d_{d}^{-\alpha_d} 	({ \;  	v+\lambda }) \right)} }}-y^2 \left( {p_u +(\frac{\lambda-\lambda^\prime}{\lambda^\prime})P_{r}(b)+C}\right)\\
{\rm	s.t.}\; & {\bf C1} .
\end{split}
\label{eq:OPtProb1Var1in_dAlgo1}
\end{align}
For a given $\lambda$, in each iteration $j$, $y^*_j$ can be found in closed-form as  $y_{j}^{\star}=\frac{\sqrt{B{\log}_2\left(1+V	d_{u}^{-\alpha_u} d_{d}^{-\alpha_d} 	({ \;  	v+\lambda }) \right)}}{ {p_u +(\frac{\lambda-\lambda^\prime}{\lambda^\prime})P_{r}(b)+C}}$. Now we solve  {\bf P3}  using \textbf{Algorithm 1} for which the convergence to the global optimal solution is proved in \cite{zappone2016achieving}. 
\begin{algorithm}[H]
	\KwData{Initialize $\lambda$, $j=1$, Maximum Iterations $J_{\max}$, Error tolerance $\epsilon$,  $Q(\lambda^j)$ }
	Find $y^\star_j$ by solving $Q(\lambda^j)$ and set $j=2$;
	
	\While{$ \vert y_{j-1}^\star- y_{j}^\star\vert\ge  \epsilon$  and $j<J_{\max}$ }{
		\begin{itemize}
			\item Update $\lambda$ by solving $Q(\lambda^j) $ for fixed $y_{j-1}^\star$ using any convex optimization tool, e.g. CVX.
			\item update $y^\star_j$
			\item $j=j+1$.
		\end{itemize} 
	}
	\KwResult{optimal desired solution $N^\star$ is then obtained from $\lambda^\star$ using  $N^\star=\frac{\lambda^\star-\lambda^\prime}{\lambda^\prime}$, }
	\caption{Optimization of Number of IRS Elements  in IRS-only Mode}
\end{algorithm}
\subsubsection{Minimization of Power Consumption Under Rate Constraint}
The problem can be formulated as:
\begin{align}
\begin{split}
{\bf P1}:	&\min_{N}\; { {p_u +NP_{r}(b)+C}}\\
{\rm	s.t.}\; & {\bf C1}: 	{B\mathrm{log}_2\left(1+V	d_{u}^{-\alpha_u} d_{d}^{-\alpha_d} 	({ \;  	v+(N+1)\lambda^\prime }) \right)}\ge R_{0}\\
& {\bf C2}:  N_{\min} \leq N \le N_{\max}.
\end{split}
\label{eq:NOptClosedFormProb}
\end{align}
The objective function in {\bf P1} is convex and monotonically decreasing w.r.t $N$ and the IRS ergodic capacity is monotonically increasing function of $N$. Therefore, the solution to the optimization problem lies at the boundary of the constraint {\bf C1}, which is given as follows:
\begin{equation}
\label{eq:NOptConstraint}
N^\star=\lceil{\frac{1}{\lambda^\prime}\left(\sqrt{\frac{\left(  2^{\frac{R_{0}}{B}}-1 \right)}{V}d_{u}^{\alpha_u} d_{d}^{\alpha_d} } -\nu -\lambda^\prime\right)}\rceil.
\end{equation}
The generalized optimal solution $N^\star$ is provided by incorporating the bound {\bf C2} as follows:
\small
\begin{equation}\label{FinalNopt}
N^\star =
\begin{cases}
N_{\min} &    	N^\star \le N_{\min}\\
N_{\max} &    	N^\star \ge N_{\max}\\
\ceil{\frac{1}{\lambda^\prime}\left(\sqrt{\frac{\left(  2^{\frac{R_{0}}{B}}-1 \right)}{V}d_{u}^{\alpha_u} d_{d}^{\alpha_d} } -\nu -\lambda^\prime\right)}
&   {\rm otherwise}
\end{cases}.
\end{equation}
\normalsize
In addition, the optimization problem \eqref{eq:NOptClosedFormProb} can be solved to optimize the variable $p_u$ given a fixed $N$. The objective function in {\bf P1} is convex and monotonically increasing w.r.t $p_u$, whereas the ergodic capacity is monotonically increasing function on $p_u$. Therefore, the solution to the optimization problem lies at the boundary of the constraint {\bf C1} and the optimal solution for $p_u^\star$ can be given as follows:
\begin{equation}
\label{eq:NOptConstraint_pu}
p^\star_{{u}}=\frac{  2^{\frac{R_{0}}{B}}-1 }{\hat{A}^2 \eta_{u}^{-1}\eta_{d}^{-1}} \frac{d_{u}^{\alpha_u} d_{d}^{\alpha_d}N_0} { \;  	v+(N+1)\lambda^\prime }.
\end{equation}

\subsection{ IRS-only Mode: Height Optimization}
Here, we maximize $\rm EE_{IRS}$ which is equivalent to maximizing the ergodic capacity $C_{\rm IRS}$ in \eqref{eq:RateMode2Approx} w.r.t height, since the IRS-only power consumption does not depend on height.
The problem can then be formulated as follows:
\begin{align}
\begin{split}
{\bf P1}:	&\max_{h}\; {\rm  	C_{IRS}}={B\mathrm{log}_2\left(1+V	d_{u}^{-\alpha_u} d_{d}^{-\alpha_d} 	({ \;  	\nu+\lambda }) \right)}\\
{\rm s.t.}\; & {\bf C1}:  h_{\min} \leq h \le h_{\max}. 
\end{split}
\label{eq:IRS_heightOpt}
\end{align}
In \eqref{eq:IRS_heightOpt}, we note that only numerator $	d_{u}^{-\alpha_u} d_{d}^{-\alpha_d} $ is a function of $h$. Therefore, to reformulate {\bf P1} we ignore the logarithm and constants in the objective function of {\bf P2} as shown below:
\begin{align}
\begin{split}
{\bf P2}:	&\max_{h}\; d_{u}^{-\alpha_u} d_{d}^{-\alpha_d} \\
\rm	s.t.\; & {\bf C1}
\end{split}
\label{eq:OPtProb1Var1in_d}
\end{align}
The optimal $h$ obtained from {\bf P2} can be substituted back in \eqref{eq:IRS_heightOpt} to obtain maximum ${\rm  	EE_{IRS}}$. 
By combining \eqref{eq:ThetaAll}, \eqref{eq:ProbLoSU}, and \eqref{eq:alphaFuncTheta}, we note  that $\alpha_u=\frac{ q_u}{1+{\varsigma_u\mathrm{exp} \left(-g_u	\arctan\left(\frac{h}{\hat{z}_{u}}\right)\right)}}+v_u$ is a function of $h$,  where  $\varsigma_u=e_u {e}^{g_u e_u}$ and $\hat{z}_{u}=|{\bf z_u}-{\bf z_s}|$.  Similarly, $\alpha_d=\frac{ q_d}{1+{\varsigma_d\mathrm{exp} \left(-g_d	\arctan\left(\frac{h}{\hat{z}_{d}}\right)\right)}}+v_d$ is a function of $h$,  where  $\varsigma_d=e_d {e}^{g_d e_d}$ and $\hat{z}_{d}=|{\bf z_u}-{\bf z_d}|$.
Clearly, the reformulated  objective function in  {\bf P2} depends on $\alpha_i, \; i \in  \{u,d\}$ which is  non-linear  due to  tangent inverse function of variable $h$ in the denominator of $\alpha_u$ and $\alpha_d$. 

Subsequently, we apply the following transformations to simplify the problem:
\begin{itemize}
	\item Taking the log of objective function of {\bf P2}, the transformed objective function becomes $$ -{\alpha_u}\log\left({d_u}\right)-{\alpha_d}\log\left({d_d}\right)={-\frac{\alpha_u}{2}}\log\left({{\hat{z}_{u}^2+h^2}}\right)-{\frac{\alpha_d}{2}}\log\left({{\hat{z}_{d}^2+h^2}}\right).$$
	\item Using $\arctan(x)\approx\frac{3x}{1+2\sqrt{1+x^2}}$, we get $$\alpha_i\approx{ q_i\left(1+{\varsigma_i\mathrm{exp}	\left(\frac{-3g_i{h}}{{\hat{z}_{i}}+2\sqrt{\hat{z}_{i}^2+h^2}} \right)}\right)}^{-1}+v_i, \quad i\in\{u,d\}$$
	\item Applying the second-order Taylor series approximation  $\exp(-x)\approx 1-x+\frac{x^2}{2}$ and some algebraic manipulations, we obtain
\end{itemize}
\begin{align} \label{eq:alpha_uApproxDetail}
\begin{split}
\alpha_i(h)\approx & \frac{ A_i\left({\hat{z}_{i}}+2\sqrt{{\hat{z}_{i}^2+h^2}}\right)^2 - B_ih\left( {\hat{z}_{i}}+2\sqrt{{\hat{z}_{i}^2+h^2}}\right)+{C_i{h^2}} }
{(1+\varsigma_i)\left({\hat{z}_{i}}+2\sqrt{{\hat{z}_{i}^2+h^2}}\right)^2	-B_i^\prime h\left({\hat{z}_{i}}+2\sqrt{{\hat{z}_{i}^2+h^2}}\right)+{C_i^\prime h^2}}, \quad i\in\{u,d\},
\end{split}
\end{align}
where		$A_i=q_i+v_i (1+\varsigma_i)$, $ B_i=3\varsigma_i v_i g_i$, $ C_i= 9/2v_i\varsigma_ig^2_i$,  $ B_i^\prime=3\varsigma_i  g_i$, $ C_i^\prime= 9/2\varsigma_ig^2_i$.
Following the above approximations, the original  \eqref{eq:OPtProb1Var1in_d} is given as follows:
\begin{align}\label{eq:Concave_Numeratr}
\begin{split}
{\bf P3}:	&\max_{h}\; 
-\frac{1}{2}\alpha_u(h)\log \left({{{\hat{z}_{u}^2+h^2}}}\right) 
{-\frac{1}{2}\alpha_d(h) \log\left({{\hat{z}_{d}^2+h^2}}\right)}\\
{\rm s.t.}\; 	& {\bf C1}.
\end{split}
\end{align}
The  mismatch in the optimal solutions is found to be negligibly small and is mainly due to the considered \texttt{arctan} and Taylor approximations, as validated in Fig.~5. 
Clearly, the problem in {\bf P3} is in the form of sum of ratio of concave-convex function as is shown in the following Proposition. This guarantees that an optimal solution for {\bf P3} can be obtained. 
\begin{prop}
	\label{Proposition4}
	The $	-\frac{1}{2}\alpha_i(h)\log \left({{{\hat{z}_{i}^2+h^2}}}\right) $ is ratio of concave-convex when
	\begin{equation}\label{ConditionConcaveMinMax }
	\hat{z}_{i} >10 \qquad\&\qquad
	\begin{cases}
	
	{\hat{z}_{i} }\ge {h^{5/4}}\left (\frac{78 A_i +14 C_i} {11B_i}\right)^{1/4} &    	 \hat{z}_{i} \ge h\\
	{h}\ge  {\hat{z}_{i}} \left ( \frac{\left( 78 A_i \hat{z}_{i} +B_i+ 14 C_i \hat{z}_{i} \right) }{12B_i} \right)^{1/4} &    h>\hat{z}_{i} 
	\end{cases}.
	\end{equation}
\end{prop}
\begin{proof}
	See   \textbf{Appendix~A}.
\end{proof}
Now, {\bf P3} can be reformulated as follows:
\begin{align}\label{eq:Convex_Numeratr}
\begin{split}
{\bf P3^\prime}:	&\min_{h}\; 
\frac{1}{2}\alpha_u(h)\log \left({{{\hat{z}_{u}^2+h^2}}}\right) 
{+\frac{1}{2}\alpha_d(h) \log\left({{\hat{z}_{d}^2+h^2}}\right)}\\
{\rm s.t.}\; 	& {\bf C1}.
\end{split}
\end{align}
Note that {\bf P3} is a  multiple-ratio fractional programming problem and can be solved by applying the  quadratic transform method, as applied earlier. 
For the sake of simplicity,
we rewrite {\bf P3$^\prime$} by using a general notation $i$, where $i={(u,d)}$, as follows \cite{shen2018fractional}:
\begin{align}\label{eq:2TermConcave_Numeratr}
\begin{split}
{\bf P4}:	&\min_{h}\; \sum_{i=(u,d)} \frac{O_i(h)}{R_i(h)}\\
{\rm s.t.}\; 	& {\bf C1},\;\;
\end{split}
\end{align}
where 	$
O_i(h)= 0.5\log({h^2+\hat{z}_{i}^2}) \times (A_i\left(\hat{z}_{i}+2\sqrt{\hat{z}_{i}^2+{h}^2}\right)^2 - B_ih ( \hat{z}_{i}+2\sqrt{\hat{z}_{i}^2+{h}^2})+{C_i{h^2}})$ and  $R_i(h)=(1+\varsigma_i)({\hat{z}_{i}}+2\sqrt{{\hat{z}_{i}^2+h^2}})^2	-B_i^\prime h\left({\hat{z}_{i}}+2\sqrt{{\hat{z}_{i}^2+h^2}}\right)+{C_i^\prime h^2}$. 
{Note that, the complexity arises due to the negative sign in $O_i(h)$ that makes $\sqrt{O_i(h)}$ a complex number. To avoid the negative sign, we rewrite the problem  {\bf P3$^{\prime}$} in the minimization form of sum of ratio of convex functions\footnote{Note that alternation is particularly applicable to the case where quadratic transform {\bf P5} is convex, not otherwise.}. Thus, after introducing auxiliary variable $y_i$ and applying quadratic transform, {\bf P5} becomes a convex problem in $h$} \cite{shen2018fractional}:
\begin{align}
\begin{split}
{\bf P5}:	&\min_{h,y_i}\;  \sum_{i=(u,d)} 	2y_i{{\sqrt{O_i(h)} }}-y_i^2 \left(R_i(h)\right)\\
{\rm	s.t.}\; 	& {\bf C1}\;\; \& \;\;
y_i\in \mathbb{R}.
\end{split}
\label{eq:2TermConcaveqUadratic}
\end{align}
For a given $h$, the optimal $y_i$ can thus be obtained  in closed
form as $y_{i}^{\star}=\frac{\sqrt{O_i(h)}}{R_i(h)}$.  The solution to the problem {\bf P5} with $Q_i(h)=2y_i{{\sqrt{O_i(h)} }}-y_i^2 R_i(h)$ in the objective function can  be obtained using \textbf{Algorithm 2} that iteratively solves the minimization problem for $h$.

Fig.~\ref{fig:IRS_only_Mode_Approx_Comparison} shows the comparison between the optimal solution obtained from solving \eqref{eq:OPtProb1Var1in_d}, \eqref{eq:Concave_Numeratr} and \eqref{eq:2TermConcaveqUadratic}, which are represented by blue, black, and red curves, respectively. Evidently, due to the considered approximations of \eqref{eq:OPtProb1Var1in_d}, the optimal solution obtained by solving   \eqref{eq:OPtProb1Var1in_d} has a slight mismatch with the exact solution obtained by solving   \eqref{eq:Concave_Numeratr} using exhaustive search method.  However, it is noteworthy that the transformation of  \eqref{eq:Concave_Numeratr} into \eqref{eq:2TermConcaveqUadratic} does not impact the optimality of the solution. 

\begin{algorithm}[H]
	\KwData{Initialize $h$, $j=1$, Maximum Iterations $J_{\max}$, Error tolerance $\epsilon$,  $Q_i(h^j)$ }
	Find $y^{\star}_{i,j}$ by solving $Q_i(h^j)$ and set $j=2$;
	
	\While{$|y^{\star}_{i,j+1}-y^{\star}_{i,j}|\ge \epsilon$,  and $j<J_{\max}$ }{
		\begin{itemize}
			\item Update $h$ by solving $Q_i(h^j) $ for fixed  $y_{i,j-1}^\star$ using any convex optimization tool e.g, CVX.
			\item update $y^{\star}_{i,j}$  $\forall i\in (u,d)$
			\item $j=j+1$.
		\end{itemize} 
	}
	\KwResult{optimal desired solution $h^\star$ }
	\caption{Height Optimization in IRS-only Mode}
\end{algorithm}

\subsection{UAV-only Mode: Height Optimization}
We formulate height optimization using 	\eqref{eq:EE_DFApprox} which is an approximation of  \eqref{eq:EE_m}  for the UAV-only mode as:		
\begin{align}\label{eq:DFEnergyEE}
\begin{split}
{\bf P1}:	&\max_{h}\; 		\rm{C}_{UAV}
=		{B\mathrm{log}_2\left( 1+\min \left(   {p_{{u}} \kappa_u \; d_u^{-\alpha_u}  \; \;\Omega_{u} },{p_{{d}} \kappa_d \; d_d^{-\alpha_d}  \; \Omega_{d} }    \right)\right) }\\ 
& {\bf C1}:  h_{\min} \leq h \le h_{\max}.
\end{split}
\end{align}
Using a similar approach followed in \eqref{eq:OPtProb1Var1in_d}, i.e. by ignoring logarithm and constant scaling function and considering only the terms that are function of $h$, we recast the optimization problem 	{\bf P1} as follows:
\begin{align}\label{eq:DFEnergyEE1}
\begin{split}
{\bf P2}:	&\max_{h}\; 	 \min \left(  I_u{d_u^{-\alpha_u} } , I_d d_d^{-\alpha_d} \right)=\max_{h}\; \min_{i}\;   I_i d_i^{-\alpha_i} \\
\rm s.t.\; 	& {\bf C1},
\;\;\;\;\;
\end{split}
\end{align}
where
$I_u={p_{{u}} \kappa_{u} \; {\Omega_{u} } }$,	$I_d={p_{{d}} \kappa_{d} \; {\Omega_{d} } }$,   $d_{u}=\sqrt{\hat{z}_{u}^2+{h}^2}$,  $d_{d}=\sqrt{\hat{z}_{d}^2+{h}^2}$, and $\alpha_i$ for $i \in (u,d )$ is given in \eqref{eq:alpha_uApproxDetail}. 
where $d_u$ and $d_d$ are convex functions of $h$, whereas $\alpha_u$ and $\alpha_d$ are ratio of concave and convex functions of $h$. Clearly, this problem is non-convex and cannot be solved directly. Therefore, we take log of  {\bf P2} which is an increasing function and does not effect the solution of the original objective. 	 {\bf P2} can then be reformulated as follows:
\begin{align}\label{eq:DFEnergyEE2}
\begin{split}
{\bf P3}:&	\max_{h}\; \min_{i}\;   \log(I_i)-\alpha_i(h) \log\left({h^2+\hat{z}_{i}^2}\right)  \\
{\rm s.t.}\; 	& {\bf C1},\;\;\;\;\;
\end{split}
\end{align}
where $\alpha_i$ is a ratio  of concave and convex functions of $h$, thus the objective function is a ratio of two functions of $h$ for $i\in (u,d)$. However, the ratio in the objective may not necessarily be concave-convex form.
{However, under a certain condition, we have proved that the objective in \textbf{P3} is indeed a concave-convex form in terms of $h$). This guarantees that an optimal solution for P3 can be obtained under specific condition.}
By substituting $\alpha_i(h)$ and simplifying the objective of {\bf P3}, we get
\begin{align}
\begin{split}\label{eq:UAVConcave_convex}
O_i(h)=&G_1 \left(\hat{z}_{i}+2\sqrt{\hat{z}_{i}^2+{h}^2}\right)^2 - G_2  h\left(\hat{z}_{i}+2\sqrt{\hat{z}_{i}^2+{h}^2}\right)
+  G_3h^2,
\end{split}
\end{align}
where 
$G_1=2\log(I_i)(1+\varsigma_i)-A_i  \log\left({h^2+\hat{z}_{i}^2}\right) $, $G_2=2\log(I_i)  B_i^\prime - B_i \log\left({h^2+\hat{z}_{i}^2}\right)  $, and
$G_3= 2\log(I_i) {C_i^\prime } - C_i\log\left({h^2+\hat{z}_{i}^2}\right)  $ and denominator function is $R_i(h)=(1+\varsigma_i)\left({\hat{z}_{i}}+2\sqrt{{\hat{z}_{i}^2+h^2}}\right)^2	-B_i^\prime h\left({\hat{z}_{i}}+2\sqrt{{\hat{z}_{i}^2+h^2}}\right)+{C_i^\prime h^2}.$ It is straight-forward to see that $R_i(h)$ is convex and in the following Proposition, we show that $O_i(h)$ in \eqref{eq:UAVConcave_convex} is a concave function of $h$ under a certain condition.
\begin{prop}
	\label{Proposition5}
	The $O_i(h)$ in \eqref{eq:UAVConcave_convex} is concave when 
	\begin{align} \label{eq:refEq}
	2 G_1 \hat{z}_{i}^3+ \left(4 G_1 + G_3\right)
	(\hat{z}_{i}^2+{h}^2)^{3/2}- G_2 h (3 \hat{z}_{i}^2 +2 h^2)  
	\end{align}
	is negative.  Using the identity that norm is less than the sum of the sides, i.e. $\sqrt{h^2+\hat{z}_{i}^2}\le h+\hat{z}_{i}$, we obtain upper bound on \eqref{eq:refEq} after simplification as
	\begin{align}
	(6 G_1+G_3) \hat{z}_{i}^3+ \left(4 G_1 + G_3-2 G_2\right) h^3+\left(4 G_1 + G_3\right) \hat{z}_{i} h^2+\left(4 G_1 + G_3- 3G_2\right) \hat{z}_{i}^2  h.
	\end{align}
	Now, for the cases $\hat{z}_{i}\ge h$ and $\hat{z}_{i}<h$ and replacing $\min(h,\hat{z}_{i})$ to $\max(h,\hat{z}_{i})$ (which gives an upper bound), we obtain the simplified condition for concavity after substituting $G_1$, $G_2$ and $G_3$ as 
	$$\log(I_i)\le \frac{(18 A_i-5B_i+4C_i)}{36(1+\varsigma_i)-10 B_i^\prime +8 C_i^\prime}  \log\left({h^2+\hat{z}_{i}^2}\right).$$
\end{prop}
Now to solve {\bf P3}, we apply  quadratic transformation available for max-min problem \cite{shen2018fractional}. The steps include recasting the problem as maximization of $z$ under the constraint on $h$ such that $z \le \frac{ O_i (h)}{R_i (h)}$. The constraint  $z \le \frac{ O_i (h)}{R_i (h)}$ can be written using quadratic transform as  $2y_i{{\sqrt{O_i(h)} }}-y_i^2 R_i(h)\ge z,\; \forall i\in (u,d)$ with ${y}_i$ as an auxiliary optimization variable. The equivalent problem of \eqref{eq:DFEnergyEE2} can then be given as:
\begin{align}\label{eq:DFEnergyEE3}
\begin{split}
{\bf P3^\prime}:&	\max_{h,{y_i}, z}\;   z  \\
& {\bf C1}\;\;\;\;\;
\&\;\;\;\;\; {\bf C2}: \small 2 y_i \sqrt{O_i(h) } - y_i^2 R_i(h) \ge z,\normalsize \;\; \forall i.
\end{split}
\end{align}	
The above problem cannot be solved due  $O_i(h)$ being the negative valued function.  To solve this, we change $z$ to $-z\ge-\frac{O_i(h)}{R_i(h)}$ to make $O_i(h)$ positive inside the  square root in {\bf C2} as follows:
\begin{align}\label{eq:DFEnergyEE4}
\begin{split}
{\bf P4}:&	\max_{h,y_i,z}\; -z  \\
& {\bf C1}\;\;\;\;\;
\&\;\;\;\;\; {\bf C2}: \small 2 y_i \sqrt{-O_i(h) } - y_i^2 R_i(h) \ge -z,\normalsize \;\; \forall i
\end{split}
\end{align}	
Now changing maximization over $h$ to minimization problem as: 
\begin{align}\label{eq:DFEnergyEE5}
\begin{split}
{\bf P5}:&	\min_{h,y_i,z}\;    z  
\\	& {\bf C1}\;\;\;\;\;
\&\;\;\;\;\; {\bf C2}.
\end{split}
\end{align}
The optimization problem is solved using {\bf Algorithm 3} for UAV only mode.

\begin{algorithm}[H]
	\KwData{Initialize $h$, $z$, $j=1$, Maximum Iterations $J_{\max}$, Error tolerance $\epsilon$,  $O_i(h^j)$ \& $R_i(h^j)$  }
	Find $y^{\star}_{i,j}$ by solving $y_{i,j}^{\star}=\frac{\sqrt{O_i(h_j)}}{R_i(h_j)}$ and set $j=2$;

	\While{$|y^{\star}_{i,j+1}-y^{\star}_{i,j}|\geq \epsilon$,  and $j\le J_{\max}$ }{
		\begin{itemize}
			\item Update $h$ and $z$ by solving  \eqref{eq:DFEnergyEE5} for fixed  $y_{i,j-1}^\star$ using any convex optimization tool, e.g. CVX.
			\item update $y^{\star}_{i,j}$  $\forall i\in (u,d)$
			\item $j=j+1$.
		\end{itemize} 
	}
	\KwResult{optimal desired solution $h^\star$ }
	\caption{Height Optimization in UAV-only Mode}
\end{algorithm}

Fig.~\ref{fig:UAV_only_Mode_Approx_Comparison} shows the comparison between the optimal solution obtained from solving \eqref{eq:DFEnergyEE}, \eqref{eq:DFEnergyEE2}, and \eqref{eq:DFEnergyEE5}, which are represented by blue, black, and red curves, respectively. Clearly, due to the considered approximations the optimal solution obtained by solving   \eqref{eq:DFEnergyEE2} has a slight mismatch with the exact solution obtained by solving  \eqref{eq:DFEnergyEE} using exhaustive search method.  However, it is noteworthy that  the transformation of  \eqref{eq:DFEnergyEE2} into \eqref{eq:DFEnergyEE5} does not impact the optimality of the solution.

\subsection{ Mode Selection to Maximize Energy Efficiency }
In this section, we derive the probabilities of selecting modes (UAV-only, IRS-only, integrated UAV-IRS) to maximize the energy efficiency. { However, first we would like to clarify that the denominator (i.e. power consumption) of energy efficiency in integrated UAV-IRS mode will always be higher than the power consumption in UAV-only and IRS-only modes.} The reason is that the power consumption of the integrated UAV-IRS mode (the sum of the power consumption of UAV-only and IRS-only modes)
is always higher than the power consumption of the UAV-only and IRS-only modes. Furthermore, the numerator which is ergodic capacity in \eqref{eq:RateMode2} chooses between the maximum SNR of either IRS-only mode or UAV-only mode. As such, the integrated UAV-IRS mode (which is optimal when the objective is to maximize the rate) is not selected when the objective is to maximize energy efficiency. Therefore, the mode selection is essentially performed between UAV-only and IRS-only modes.
In what follows, we derive the mode selection probabilities given the instantaneous fading channels and devise	a criterion to select how many active IRS elements are needed to  maximize energy efficiency in IRS-only mode. We use the proposed criterion for mode selection and obtain optimal heights  in above subsections to maximize the overall energy efficiency of the integrated UAV-IRS system.

The probability of selecting IRS-only mode can be formulated as follows:
\begin{align}\label{eq:CritIRS}
\begin{split}
\mathbb{P}_{\rm IRS}=\Pr\left(\rm \Gamma_{\rm IRS } \ge \frac{\Gamma_{\rm UAV} P_{\rm IRS }}{P_{\rm UAV}}\right)=1-\Pr\left(\rm \Gamma_{\rm IRS } < \frac{\Gamma_{\rm UAV} P_{\rm IRS }}{P_{\rm UAV}}\right).
\end{split}
\end{align}
Conditioned on $\Gamma_{\rm UAV}$, the probability in \eqref{eq:CritIRS} can be derived as follows:
\begin{align}\label{eq:Crit2}
\begin{split}
\mathbb{P}_{\rm IRS}= \mathbb{E}_{\Gamma_{\rm UAV}}\left[1-F_{\Gamma_{\rm IRS} }\left(\frac{ {\Gamma_{\rm UAV} P_{IRS}}}{ P_{\rm UAV}}\right)\right] 	\stackrel{(a)}{=}1-\int_{0}^{\infty}F_{\Gamma_{\rm IRS} }\left(\frac{ {\Gamma_{\rm UAV} P_{IRS}}}{ P_{\rm UAV}}\right) \;f_{\Gamma_{\rm UAV}}(z) dz, 
\end{split}
\end{align}
where $F_{\Gamma_{\rm IRS} }\left(\frac{ {\Gamma_{\rm UAV} P_{\rm IRS}}}{ P_{\rm UAV}}\right)$ is obtained by replacing $\Gamma_0$ with $ \frac{ {\Gamma_{\rm UAV} P_{\rm IRS}}}{ P_{\rm UAV}}$ in \eqref{1IRS}. The density function of $\Gamma_{\rm UAV}$ in \eqref{eq:SNRDFMode} is obtained by using order statistics and differentiating \eqref{eq:Pout2Rice} as $f_{\Gamma_{\rm UAV}}(z)=(1-F_{X_u}(z)) f_{X_d}(z)+(1-F_{X_d}(z)) f_{X_u}(z)$, where  $f_{X_i}(z)$ and $F_{X_i}(z)$ are given in \eqref{eq:LOSdistribution} and \eqref{eq:aaaa}, respectively. Subsequently, the probability of UAV-only mode selection can be given as $\mathbb{P}_{\rm UAV}=1-\mathbb{P}_{\rm IRS}$.

Now, to maximize the energy efficiency at an arbitrary height, we design the following mode selection criterion based on the average SNR\footnote{Generally, the instantaneous CSI may not be available at the receiver.} to select the IRS-only mode, i.e.
\begin{align}\label{eq:EEBasedCriteria}
\begin{split}
&\mathbb{E}[\Gamma_{\rm IRS}] \ge \frac{{\mathbb{E}[\Gamma_{\rm UAV}]} P_{\rm IRS}}{P_{\rm UAV}}. 
\\& N>\frac{\left( p_u-P_r(b)+C\right)  \min \left(p_u \kappa_u d_u^{-\alpha_u}\Omega_u,p_d \kappa_d d_d^{-\alpha_d}\Omega_d \right) -\nu   (p_u+p_d+C) V d_u^{-\alpha_u}d_d^{-\alpha_d}}{\lambda^\prime\left( p_u+p_d+C\right)\left( Vd_u^{-\alpha_u}d_d^{-\alpha_d}  \right)-P_r(b)  \min \left(p_u \kappa_u d_u^{-\alpha_u}\Omega_u,p_d \kappa_d d_d^{-\alpha_d}\Omega_d \right)} -1 = N_{\rm th}.
\end{split}
\end{align}
That is, the number of IRS elements should be greater than $N_{\rm th}$ to enable the IRS-only mode. 
Another way to maximize the energy efficiency is to calculate $\frac{\mathbb{E}[\Gamma_{IRS}]}{P_{\rm IRS}}$ and $\frac{\mathbb{E}[\Gamma_{UAV}]}{P_{\rm UAV}}$ with their optimal heights calculated in Section~V.B ({\bf Algorithm~2}) and Section~V.C ({\bf Algorithm~3}), respectively. Then choose the mode and optimal height corresponding to whichever term becomes maximum.

{\bf Remark:} For mode selection  based on the power consumption, the integrated UAV-IRS mode will never be selected due to its higher power consumption compared to the UAV-only and IRS-only modes. Furthermore, IRS-only mode will be selected  when ${N }\le \frac{p_d }{P_r(b)}$ and vice versa for the UAV-only mode. Similarly, for the SNR-based mode selection, then integrated UAV-IRS mode will always be selected as it chooses the maximum SNR of the IRS-only and UAV-only modes.

\section{Numerical Results and Discussion}
\label{Sec:Simulation}

In this section, we verify the accuracy of our derived expressions  and obtain insights related to the  number of IRS elements and the optimal height of UAV for different communication modes.  Unless stated otherwise, the simulation parameters are: the maximum distance the UAV can travel  $D=2000$~m, $B=5$~MHz, $H=350$~m, $p_u=p_d=50$~dBm, $\eta_u=0.009$, $\eta_d=0.01$,  $\Gamma_0=8$dB, $q_u=q_d=-1.5$,  $v_u=v_d=3.5$, $w_u=w_d=15$ dB, $z_u=z_d=5$, $D_{\rm IRS}=0.5$~m. We use $N_0=10^{-17}$W/Hz \cite{sekander2018multi}, that justifies the values we use for $\gamma_u$ and $\gamma_d$ herein.
\begin{figure*}[t]
	\begin{minipage}{0.48\textwidth}
		\includegraphics[width=10cm, height=9cm]{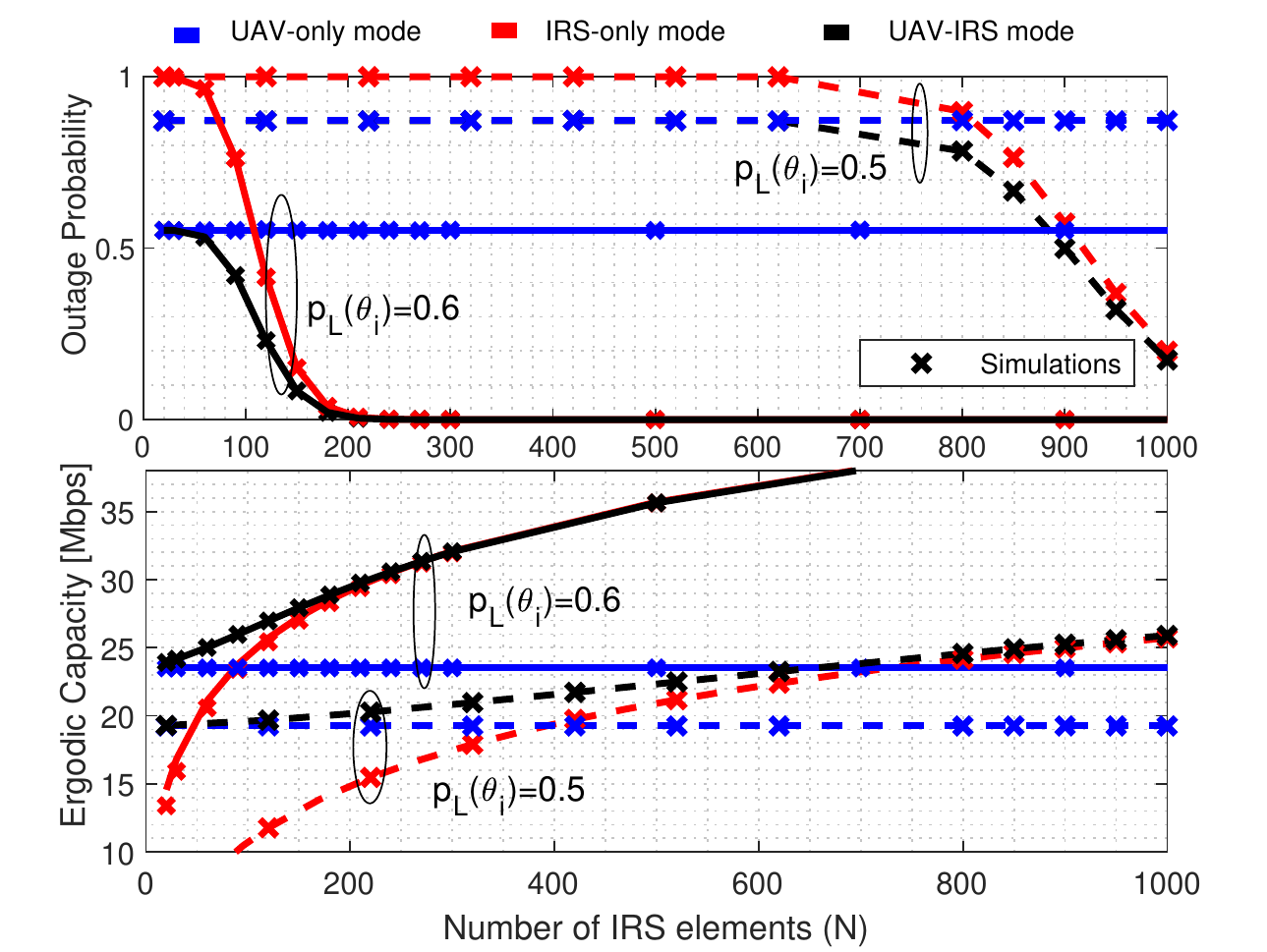}
		\caption{Performance comparison of outage probability, and ergodic capacity for IRS-only, UAV-only and integrated UAV-IRS mode for  $d=750$~m,  $D=5000$~m, $p_u=p_d=55$~dBm, $P_r(b)=108\times10^{-2}$W, $R_{\rm SI}=5$dB, $E_b/N_0=122$~dB, and $\Gamma_0=15$ dB.
		}
		\label{fig:EE_rate_Outage_UAV_IRSVsNVectorAllDrHina1}	
	\end{minipage}\hfill
	\begin{minipage}{0.48\textwidth}
		\includegraphics[width=10cm, height=9cm]{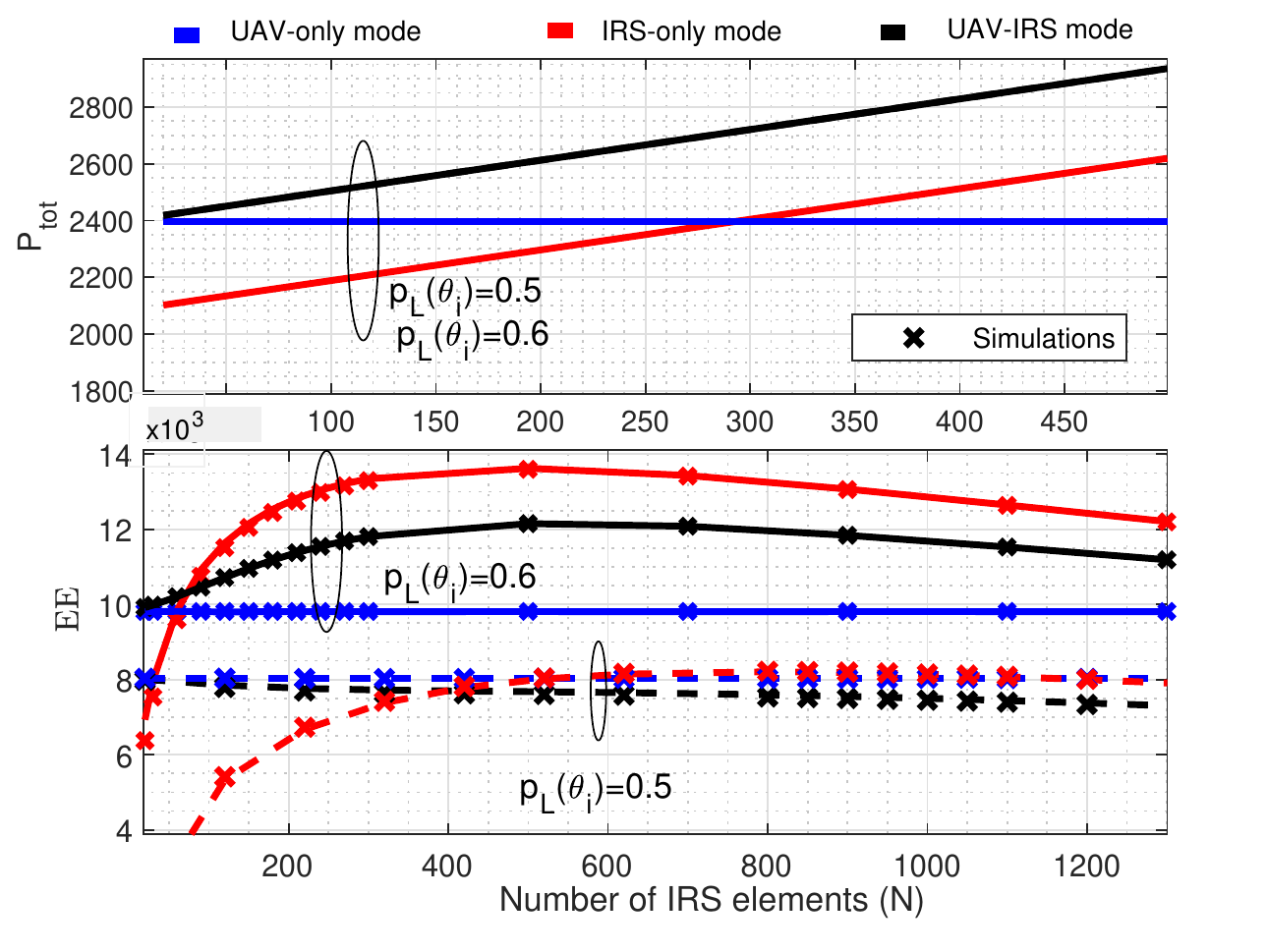}
		\caption{Performance comparison of power consumption, and energy-efficiency for IRS-only, UAV-only and integrated UAV-IRS mode for  $d=750$~m,  $D=5000$~m, $p_u=p_d=55$~dBm, $P_r(b)=108\times10^{-2}$W, $R_{\rm SI}=5$dB, $E_b/N_0=122$~dB, and $\Gamma_0=15$ dB.
		}
		\label{fig:EE_rate_Outage_UAV_IRSVsNVectorAllDrHina2}	
	\end{minipage}\hfill
\end{figure*}

\begin{figure*}[t] 
	\begin{minipage}{0.48\textwidth}
		\includegraphics[width=9cm, height=9cm]{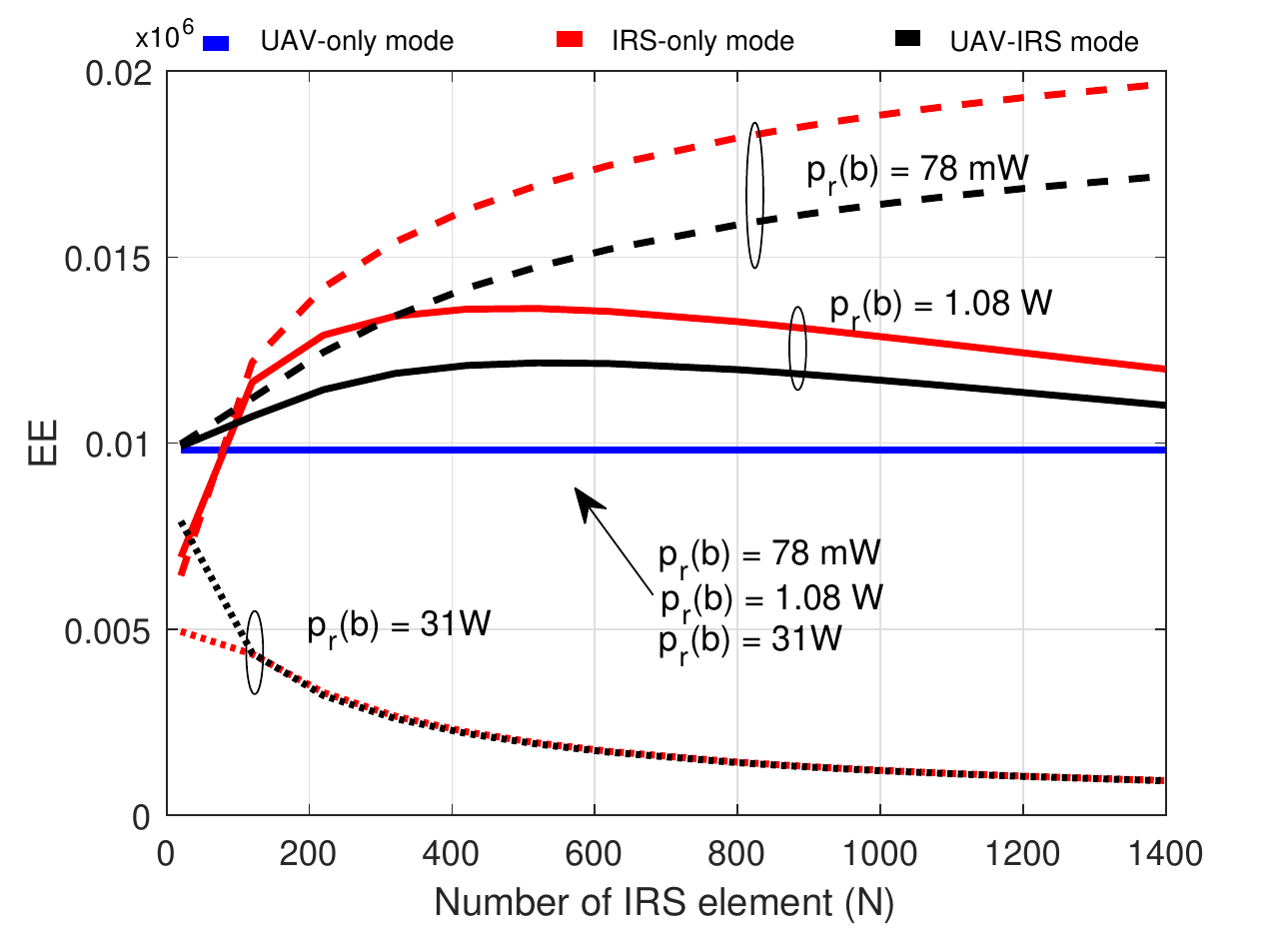}	
		\caption{ Power consumption and energy-efficiency w.r.t $N$ for different different bit resolution $P_r(b)$ for $d=750$m, and $D=5000$m $p_u=p_d=55$ dBm,  $R_{\rm SI}$ = 55, $E_b/N_0$ = 122 dB, and $\Gamma_0=15$dB.	}
		\label{fig:EFig9EVsPrb}	
	\end{minipage} \hfill
	\begin{minipage}{0.5\textwidth}
		\includegraphics[width=9cm, height=9cm]{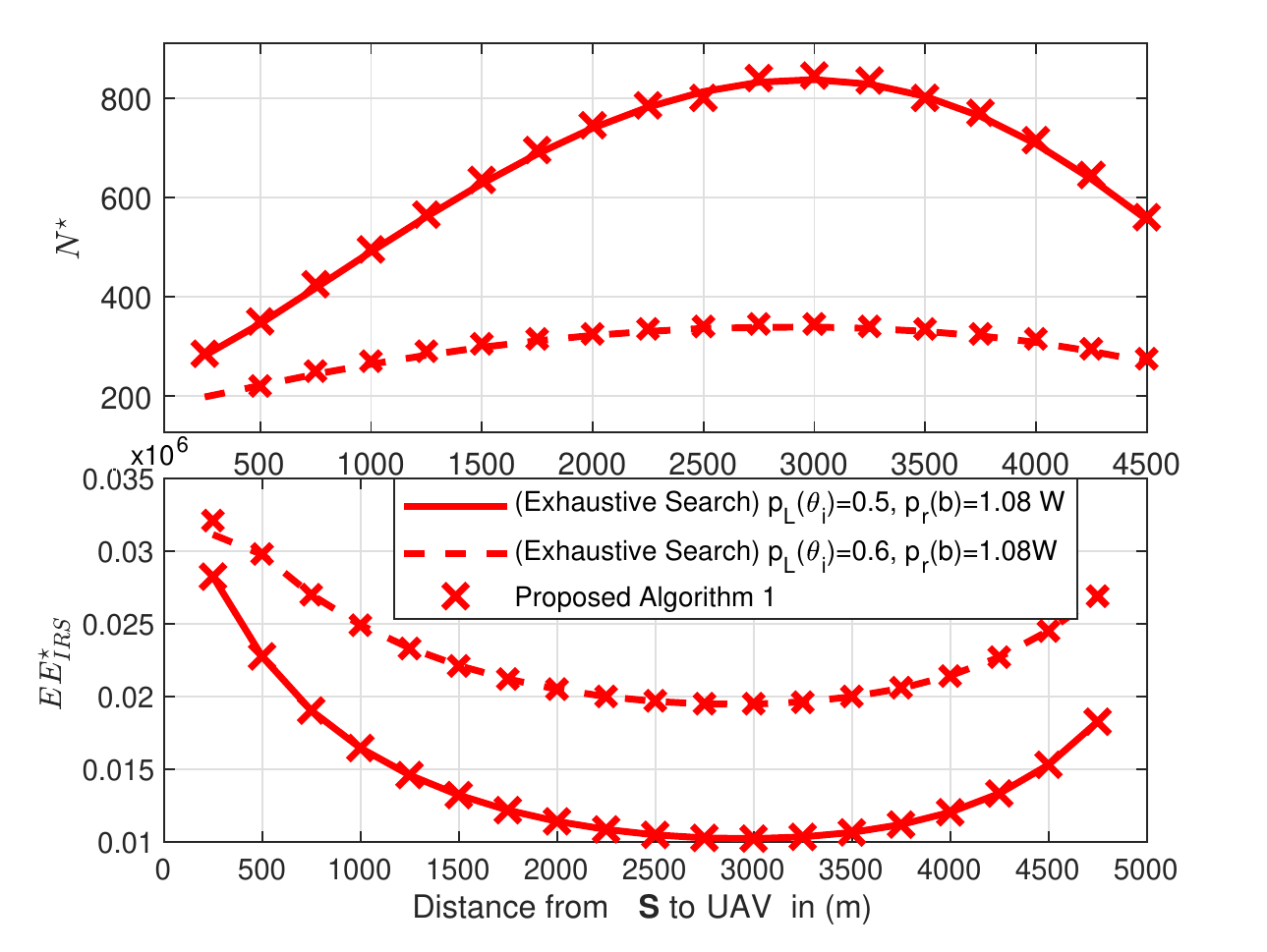}	
		\caption{Optimal number of IRS elements and the optimal $EE^\star_{\rm IRS}$ comparison for different source to UAV distance for the  provided in Fig.~\ref{fig:EE_rate_Outage_UAV_IRSVsNVectorAllDrHina1} and Fig.~\ref{fig:EE_rate_Outage_UAV_IRSVsNVectorAllDrHina2} and different bit resolution $P_r(b)$.	}
		\label{fig:MAxEE_OptimalN_ForDistance}	
	\end{minipage}
\end{figure*}

Fig.~\ref{fig:EE_rate_Outage_UAV_IRSVsNVectorAllDrHina1} compares the  outage probability and ergodic capacity w.r.t the number of IRS elements for the UAV-only, IRS-only and integrated UAV-IRS modes. Clearly,  UAV-only mode is independent of  $N$. However, as $N$ increases, the IRS-only mode and the integrated UAV-IRS mode minimize the outage probability and maximize the capacity due to enhanced IRS transmission link. For larger values of $N$, the IRS-only transmissions become strong and the opportunistic selection between the UAV-only and IRS-only modes improves the performance of integrated UAV-IRS mode.  As expected, the integrated UAV-IRS mode outperforms the IRS-only and UAV-only mode for all $N$ in terms of outage  and ergodic capacity. An interesting observation is that the lower LoS probability worsens the performance of all schemes. That is, a higher value of $N$ is needed to minimize the outage and maximize the transmission capacity for scenarios with lower LOS.

Fig.~\ref{fig:EE_rate_Outage_UAV_IRSVsNVectorAllDrHina2} compares the power consumption and energy-efficiency w.r.t the number of IRS elements for the UAV-only, IRS-only, and integrated UAV-IRS modes.
Clearly, the power consumption and energy efficiency of UAV-only mode do not depend on $N$. However, for the IRS-only mode, the power consumption increases with $N$ and the slope keeps increasing with the value of the power consumption per IRS element $P_r(b)$. Note that the power consumption does not change with the LoS probability; therefore the reduction in energy efficiency with the decrease in LoS probability is only due to the reduction in transmission capacity.  Furthermore,  the energy efficiency  first increases up to a certain value of $N$,  because the capacity is dominant than power consumption in this regime. Later, for larger values of $N$, the power consumption becomes dominant and thus the reduction in energy efficiency is evident.   Finally, it is intuitive to see that the power consumption of the integrated UAV-IRS mode is higher than the other modes; therefore, an efficient mode selection mechanism is important.

Fig.~\ref{fig:EFig9EVsPrb} shows the effect of power consumption of bit resolution $ P_r(b)$ on the energy efficiency of the three communication modes. It is clear that the UAV-only mode is independent of $P_r(b)$. However, the IRS-only and integrated UAV-IRS modes show that an optimal number of IRS elements exists which increases with the reduction in $P_r(b)$. In particular, for smaller values of $P_r(b)$, the EE continues to increase for a wider range of $N$, because the increase in $N$ does not significantly increase the power consumption, whereas the capacity keeps increasing. For higher values of $P_r(b)$, the power consumption of IRS elements becomes more dominant than the impact of IRS elements on  the ergodic capacity. As such, after a specific value of $N$, a  decreasing energy-efficiency trend can be observed. Clearly, for very high values of $P_r(b)$, minimizing IRS elements is necessary to maximize energy efficiency. Similar trends are observed for EE in integrated UAV-IRS mode with lower gain than the IRS-only mode, because this mode consumes more power then the IRS-only and UAV-only modes.


\begin{figure*}[t]
	
	\begin{minipage}{0.48\textwidth}
		\includegraphics[width=10cm, height=10cm]{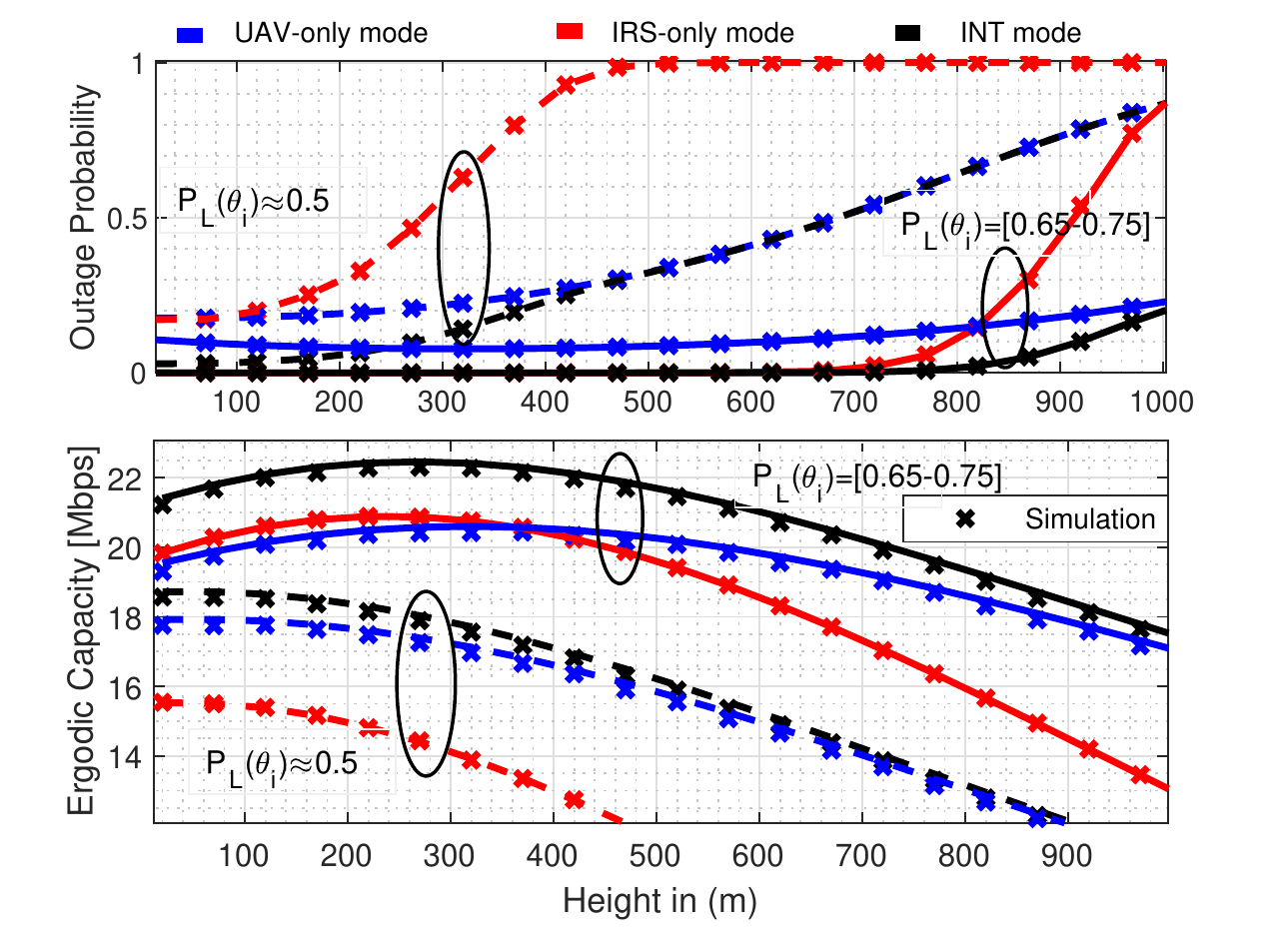}
		\caption{ 
			Performance comparison of outage probability and ergodic capacity for IRS-only, UAV-only and integrated UAV-IRS mode for  $d=1050$~m,  $D=2000$~m, $p_u=p_d=55$~dBm, $N=270$, $P_r(b)=108$ mW, $R_{\rm SI}=45$dB w.r.t height of UAV. 
		}
		\label{fig:EE_rate_Outage_UAV_IRSVsHeighteAll1}	
	\end{minipage} \hfill
	\begin{minipage}{0.48\textwidth}
		\includegraphics[width=10cm, height=10cm]{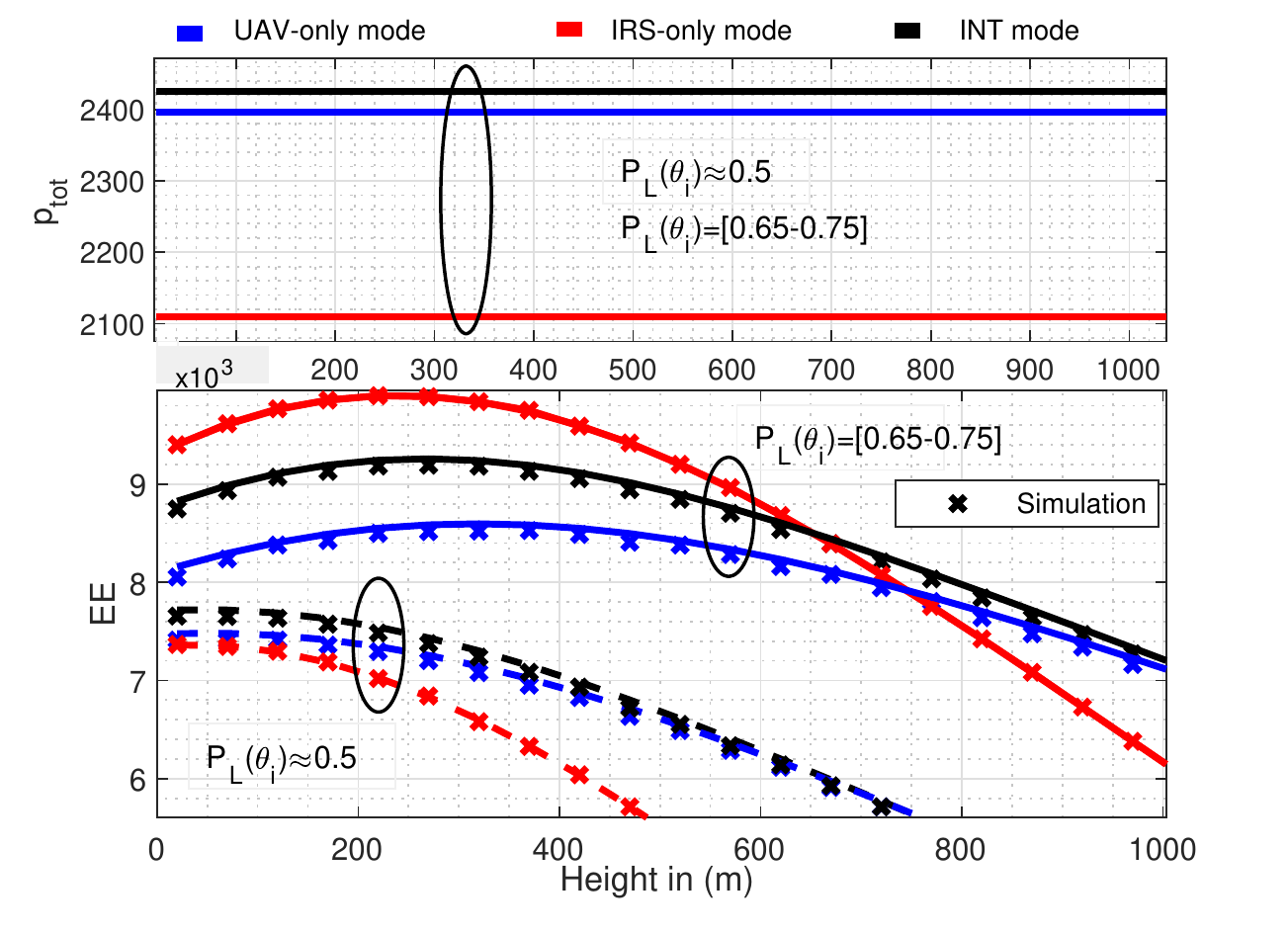}
		\caption{Performance comparison of power consumption and EE for IRS-only, UAV-only and integrated UAV-IRS mode for  $d=1050$~m,  $D=2000$~m, $p_u=p_d=55$~dBm, $N=270$, $P_r(b)=108m$W, $R_{\rm SI}=45$dB w.r.t height of UAV. 
		}
		\label{fig:EE_rate_Outage_UAV_IRSVsHeighteAll2}	
	\end{minipage}
\end{figure*}

Fig.~\ref{fig:MAxEE_OptimalN_ForDistance} shows the optimal number of IRS elements $N^\star$ (obtained using \textbf{Algorithm~1}) continues to increase as a function of the distance between the source and UAV. However, the corresponding values of optimum energy efficiencies continue to decrease with the increasing distance between the source and UAV.
On the other hand, when the distance from the source to UAV decreases, higher values of optimum energy-efficiency can be achieved with less number of IRS elements. This trend is also true when the distance from UAV to destination decreases. {The proposed optimal solution (shown by marker) matches well with the optimal solutions obtained by an exhaustive search. }
Furthermore, we note that a low LoS probability $p_L(\theta_i)=0.5$ requires more IRS elements for optimal function while the maximum energy-efficiency values obtained are still low. On the other hand, when $p_L(\theta_i)=0.6$, a fewer number of IRS elements provide higher optimum energy efficiency values. 
In summary, we can conclude that if bit resolution power is very small,  then using maximum  number of IRS elements is optimal, whereas when the bit resolution power is significantly large, then using minimum number of IRS elements is optimal.


\begin{figure*}[t]
	\begin{minipage}{0.5\textwidth}
		\includegraphics[width=9cm, height=5.5cm]{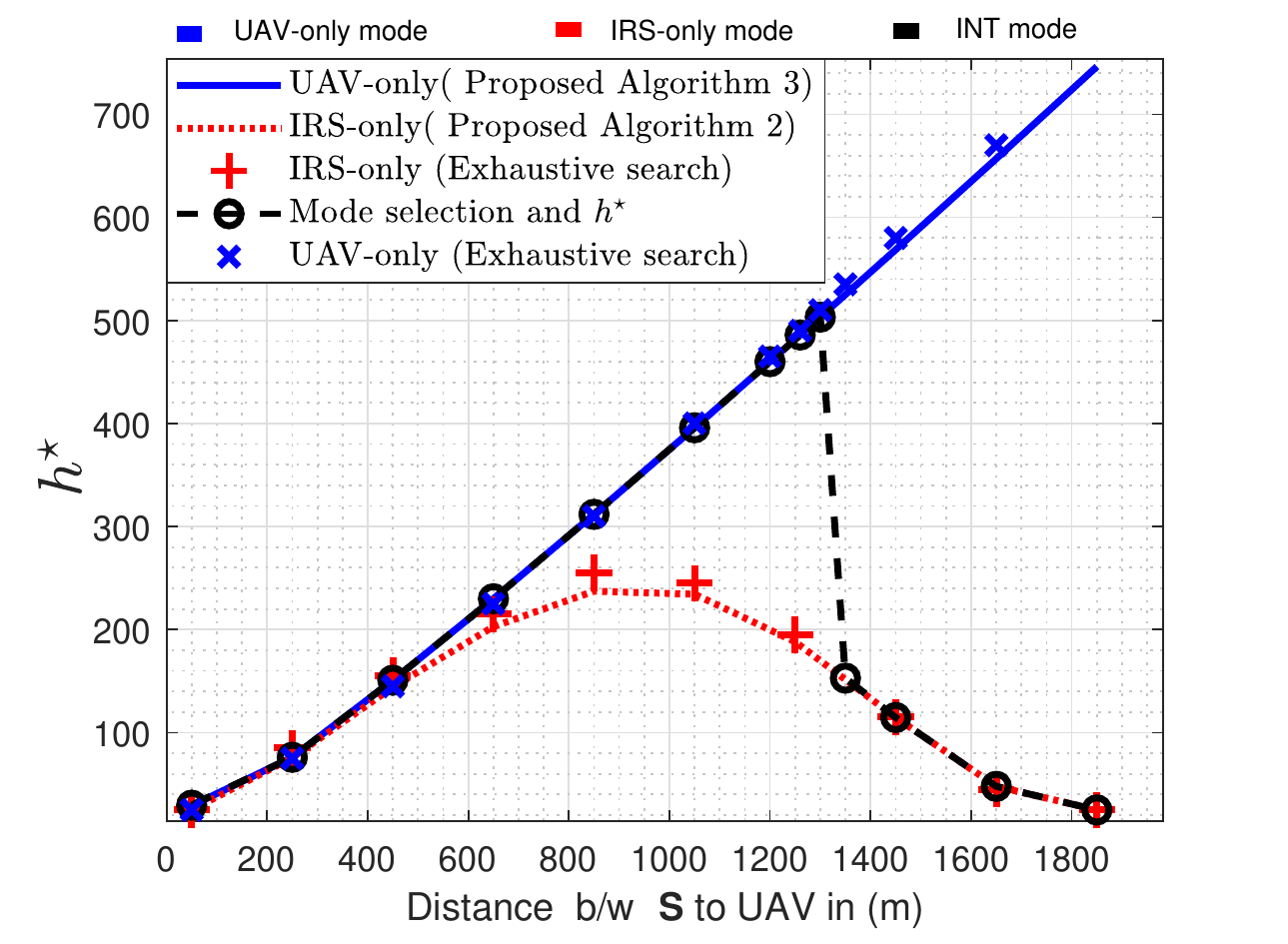}
		\caption{ Optimal  UAV height  for  different  source  to  UAV  distance for  $E_b/N_0=130$~dB, $R_{\rm SI}=38$~dB, $N=30$,  $P_r(b)=1.08$~W.
		}
		\label{fig:MAxEE_OptimalHeight_ForDistanceDrHina1}	
	\end{minipage}\hfill
	\begin{minipage}{0.48\textwidth}
		\includegraphics[width=9cm, height=5.5cm]{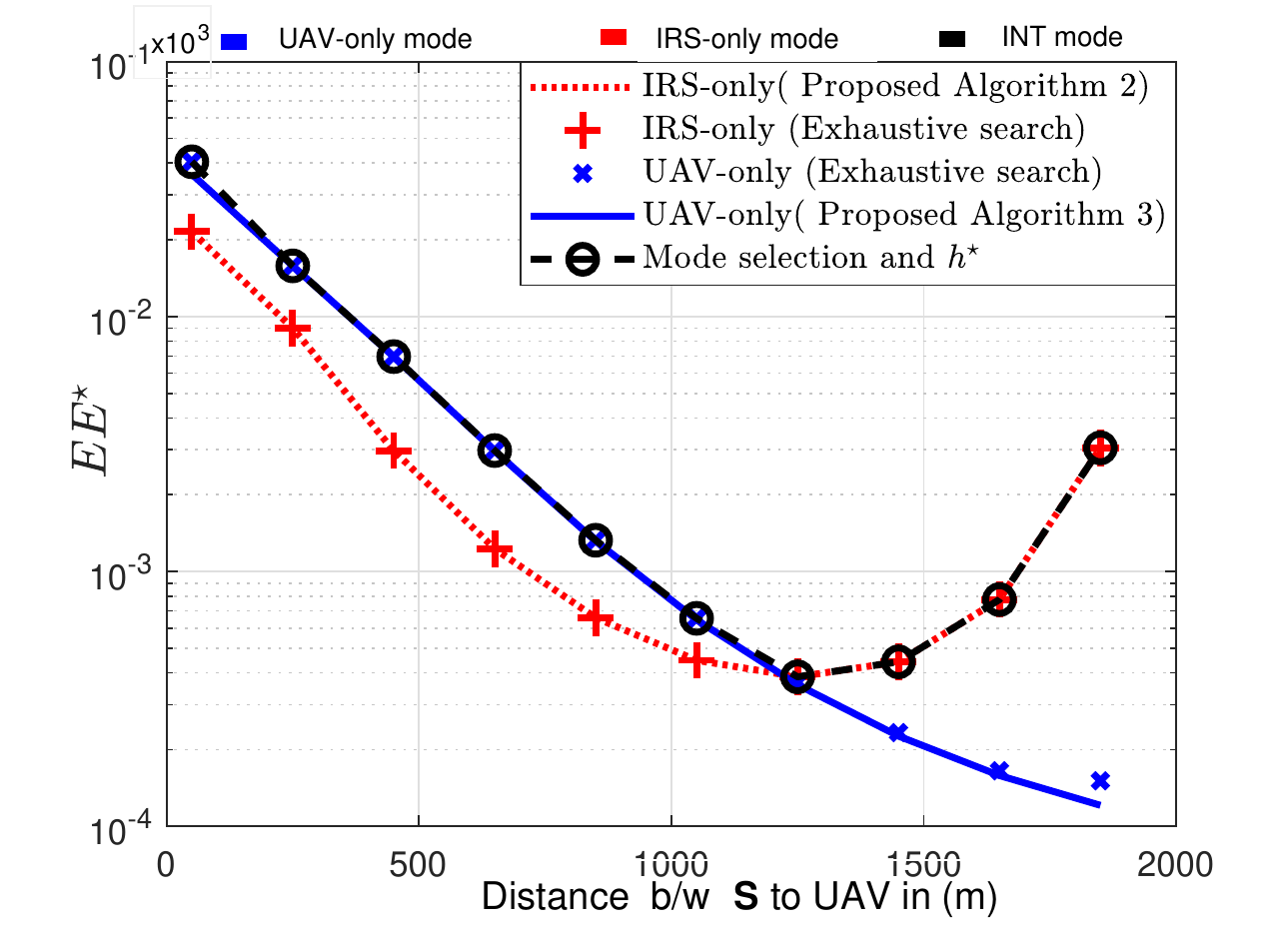}
		\caption{  Optimal  EE  for  different  source  to  UAV  distance for $E_b/N_0=130$~dB, $R_{\rm SI}=38$~dB, $N=30$, $P_r(b)=1.08$~W.
		}
		\label{fig:MAxEE_OptimalHeight_ForDistanceDrHina2}	
	\end{minipage}
\end{figure*}
Fig.~\ref{fig:EE_rate_Outage_UAV_IRSVsHeighteAll1} shows the the outage probability and ergodic capacity versus height of the UAV considering the UAV-only, IRS-only, and integrated UAV-IRS modes. We note that the optimal height varies depending on the selected communication mode. Starting with the outage probability, for weak LoS, we have a higher outage probability in general. However, we note that for weak LoS $P_L(\theta_i)\approx 0.5$, the UAV-only mode outperforms the IRS-only mode, and as expected, the integrated mode performs better than both modes. However, for strong LoS, the IRS-only mode performs better than the UAV-only mode for a wide range of heights. Similar trends can also be seen from the ergodic capacity and that the IRS-only mode dominates the UAV-only mode for smaller heights and the UAV-only mode performs better for higher altitude. 

Fig.~\ref{fig:EE_rate_Outage_UAV_IRSVsHeighteAll2} depicts power consumption and energy efficiency performance with respect to height. The power consumption is independent of the LoS probability and height of UAV. The energy-efficiency in strong LoS  $P_L(\theta_i)=[0.65~0.75]$ outperforms the EE in weak LoS  $P_L(\theta_i)=0.5$. For weak LoS, $EE_{\rm IRS}$ was least energy-efficient. However, for strong LoS, the IRS-only mode becomes the most energy-efficient mode for a wide range of UAV altitudes, since the IRS only mode power consumption is much lower than the other communication modes (i.e. small $P_r(b)$ and $N$).

Fig.~ \ref{fig:MAxEE_OptimalHeight_ForDistanceDrHina1} compares the optimal heights for  different distances between source and UAV. This figure shows that the height calculated from the proposed \textbf{Algorithm~2} and \textbf{3} matches well with the exact optimal height obtained from exhaustive search. In addition, the performance of analytical mode selection criterion and its corresponding optimal height can also be seen. This shows that for the distance between source and UAV less then 1200~m, the UAV-only mode is optimal, whereas when the UAV is close to the destination the IRS-only mode is optimal. Hence, the optimal height switches to IRS-only height. The same trend is also observed from Figure~\ref{fig:MAxEE_OptimalHeight_ForDistanceDrHina2} which represents the optimal energy efficiency vs distance between {\bf S} and UAV and follows the same trend as in Fig.~\ref{fig:MAxEE_OptimalHeight_ForDistanceDrHina1}.

\section{Conclusion}
\label{sec:Conclusion}
We have analyzed the end-to-end performance in terms of SNR outage probability, ergodic capacity, and energy efficiency for an integrated UAV-IRS relaying system that can operate in three different modes, namely, IRS-only mode, UAV-only mode and integrated UAV-IRS mode. 
For the IRS-only mode, we  optimized the number of IRS elements and UAV height, whereas we have optimized the UAV height for the UAV-only mode. We have observed that the optimal height varies based on the selected transmission mode. We have also provided an analytical criterion for optimal height and mode selection in terms of energy efficiency. 


\appendices
\renewcommand{\theequation}{A.\arabic{equation}}
\setcounter{equation}{0}
\section*{Appendix A: ratio of Concavity-convexity of \eqref{eq:Concave_Numeratr}}
We write $ -\frac{1}{2}\alpha_i(h)\log \left({{{\hat{z}_{i}^2+h^2}}}\right) =\frac{-O_i(h)}{R_i(h)}$.	The numerator  $-O_i(h)$ is concave  when 
the second derivative is $-\frac{d^2 O_i(h)}{dh^2}\le 0$. This is true if,
\small
\begin{align}
\label{A2}
\begin{split}
-&\left[ 4A_i \hat{z}_{i}^5 +\hat{z}_{i}^4  \left(-2 B_i+5 A_i  \sqrt{ \hat{z}_{i}^2+h^2}\right) +  \hat{z}_{i}^3 \left(8A_i h^2- B_i \sqrt{ \hat{z}_{i}^2+h^2}\right)  +  h^2 \hat{z}_{i}^2\left( -4B_i +5 (3A_i+C_i) \sqrt {\hat{z}_{i}^2+h^2}  \right )\right. \\ &\left.+ h^4 \left( -2B_i+3(4A_i+C_i) \sqrt {\hat{z}_{i}^2+h^2}   \right) + \hat{z}_{u} \left( 4A_i h^4+B_i h^2 \sqrt{\hat{z}_{i}^2+h^2}\right) + ({\hat{z}_{i}^2+h^2})\log (\hat{z}_{i}^2+h^2)\right.
\\  &\times \left.  \left(    2 A_i  z^3_{us}+(4A_i+C_i) h^2 \sqrt{\hat{z}_{i}^2+h^2}  +  \hat{z}_{i}^2\left(-B_i+(4A_i +C_i)\sqrt{\hat{z}_{i}^2+h^2} \right)  \right) \right] \le 0.
\end{split}
\end{align}
\normalsize
Starting from \eqref{A2}, we use $ \log (\hat{z}_{i}^2+x^2) \ge \log(\hat{z}_{u}) +\frac{x^2}{x^2+\hat{z}_{i}^2}$ and $\log(\hat{z}_{i})\ge 1$ when $\hat{z}_{i}\ge 10$ and $\min(\hat{z}_{i})>10$ which shows that source and UAV should be at least 10m distance apart in the horizontal plane (which gives one of the condition to prove concave numerator). Under this condition, we obtain,
\small	\begin{align}
\begin{split}
& -\left[  6 \hat{z}_{i}^5 A_{i} + 12 A_{i} h^2\hat{z}_{i}^3  + 
5h^4   (4 A_{i} + C_{i}) \sqrt{\hat{z}_{i}^2 + h^2} + 
\hat{z}_{i}^4  (9 A_{i} + C_{i}) \sqrt{\hat{z}_{i}^2 + h^2} + 
\hat{z}_{i}^2 h^2  (27 A_{i} + 8 C_{i})\sqrt{\hat{z}_{i}^2 + h^2} + \right.\\&\left.
\hat{z}_{i} (4 A_{i} h^4 + B_{i} h^2 \sqrt{\hat{z}_{i}^2 + h^2}) - \hat{z}_{i}^3 B_{i} \sqrt{\hat{z}_{i}^2 + h^2}-2 B_{i} h^4  -3 B_{i}  \hat{z}_{i}^4 -6 B_{i} \hat{z}_{i}^2 h^2 \right] \le 0.
\end{split}
\end{align}
\normalsize
Substituting the lower bound $\sqrt{\hat{z}_{i}^2+h^2}\ge \max(\hat{z}_{i},h)$, which does not change the negativity of the expression,  we have
\small
\begin{align}\label{eq:refbothcase}
\begin{split}
&   6 \hat{z}_{i}^5 A_i + 12 A_i h^2\hat{z}_{i}^3  + 
5h^4   (4 A_i + C_i) \max(\hat{z}_{i},h) + 
\hat{z}_{i}^4 ( (9 A_i + C_i)\max(\hat{z}_{i},h)) + 
\hat{z}_{i}^2 h^2  (27 A_i + 8 C_i) \times \\&\;\; \max(\hat{z}_{i},h) + 
\hat{z}_{i} (4 A_i h^4 + B_i h^2 \max(\hat{z}_{i},h)) - \hat{z}_{i}^3 B_i\max(\hat{z}_{i},h)-2 B_i h^4  -3 B_i  \hat{z}_{i}^4 -6 B_i \hat{z}_{i}^2 h^2 \ge 0.
\end{split}
\end{align} 
\normalsize
To simplify the expression, we consider
{\bf case (i)} when $\hat{z}_{i}>h $, and substitute $\max(\hat{z}_{i},h)= \hat{z}_{i}$ that yields:
$$
(15 A_{i} + C_{i}) \hat{z}_{i}^5  + 
(24 A_{i} + 5 C_{i})  \hat{z}_{i} h^4 + 39 A_{i} h^2\hat{z}_{i} ^3  + 8 C_{i} h^2\hat{z}_{i} ^3  -4 B_{i} \hat{z}_{i} ^4  - 5 B_{i} \hat{z}_{i}^2  h^2 - 2 B_{i} h^4 \ge 0.$$
{Replacing $\hat{z}_{i}$ by $h$ in the positive terms and $h$ by $\hat{z}_{i}$ in negative terms, we get} 	 
\begin{equation}
11B_{i} \hat{z}_{i} ^4- (78 A_{i} +14 C_{i})h^5 >0 
\implies
{z_{i} }\ge {h^{5/4}}\left (\frac{78 A_{i} +14 C_{i}} {11B_{i}}\right)^{1/4}.
\end{equation} 
\normalsize
Similarly, for {\bf case (ii): when $h>\hat{z}_{i} $:} we substitute $\max(\hat{z}_{i},h)\ge h$ in \eqref{eq:refbothcase}, replacing $h$ by $\hat{z}_{i}$ in the positive terms and $\hat{z}_{i}$ by $h$ in negative terms and simplification gives
$$
=   -\left[  6 \hat{z}_{i}^5 A_{i} + 12 A_{i} \hat{z}_{i}^5  + 
5 \hat{z}_{i}^5   (4 A_{i} + C_{i}) + 
\hat{z}_{i}^5  (9 A_{i} + C_{i}) + 
\hat{z}_{i}^5 (27 A_{i} + 8 C_{i}) +
4 A_{i} \hat{z}_{i}^5 + B_{i} \hat{z}_{i}^4 - 12 B_{i} h^4 \right]\le 0		$$	 
\begin{equation}
\label{eq:condA6}
\implies
{h}\ge  {\hat{z}_{i}} \left ( \frac{ 78 A_i \hat{z}_{i} +B_i+ 14 C_i \hat{z}_{i}  }{12B_i} \right)^{1/4}.
\end{equation}
\normalsize

However, the denominator  $R_i(h)=(1+\varsigma_i)({\hat{z}_{i}}+2\sqrt{{\hat{z}_{i}^2+h^2}})^2	-B_i^\prime h\left({\hat{z}_{i}}+2\sqrt{{\hat{z}_{i}^2+h^2}}\right)+{C_i^\prime h^2}$ is convex when the second derivative of $R_i(h)$ is positive. The term $\frac{d^2 R_i(h)}{dh^2}\ge 0$ is positive when $2 \hat{z}_{i}^3 A_{i} + (\hat{z}_{i}^2 + h^2) (4 A_{i} + C_{i})  \sqrt{\hat{z}_{i}^2 + h^2}\ge 
\hat{z}_{i}^2 B_{i}  $ which is true since $2A_{i} \hat{z}_{i}>B_{i}$ because $A_i, B_i, C_i$ are order of tens but $z_i$ is in order of hundreds and thousands, hence $R_i(h)$ is convex.

Hence, (A.4) and \eqref{eq:condA6} under the constraint $\min(\hat{z}_{i})>10$ gives the condition on concavity of $-O_i(h)$.

\bibliographystyle{ieeetran}
\bibliography{IEEEabrv,Taniya_REF_2019_UAVIRS}

\end{document}